\documentclass[11pt, a4paper]{article}
\usepackage{authblk}
\usepackage{xcolor}
\usepackage{graphicx}
\usepackage[a4paper, top=0.5in, bottom=0.75in, left=0.75in, right=0.75in]{geometry}
\usepackage{placeins}
\usepackage{amssymb}
\usepackage{amsmath}
\usepackage{stmaryrd}
\usepackage{amsthm}
\usepackage{enumitem}

\usepackage{xcolor}

\usepackage{accents}
\usepackage{amssymb}
\usepackage{bm}
\usepackage{sidenotes}
\usepackage{url}

\def\orbisingular{\rotatebox[origin=c]{70}{$\prec$}}

\usepackage{hyperref}
\hypersetup{
  colorlinks = true,
  linkcolor=darkblue, citecolor=darkblue,
  urlcolor=black
}

\definecolor{darkblue}{rgb}{0.05,0.25,0.65}
\definecolor{darkgreen}{RGB}{20,140,10}
\definecolor{lightgray}{rgb}{0.9,0.9,0.9}
\definecolor{darkorange}{RGB}{200,100,5}
\definecolor{darkyellow}{rgb}{.91,.91,0}
\definecolor{lightolive}{RGB}{225, 220, 185}

\newlength{\dhatheight}
\newcommand{\doublehat}[1]{%
    \settoheight{\dhatheight}{\ensuremath{\hat{#1}}}%
    \addtolength{\dhatheight}{-0.35ex}%
    \hat{\vphantom{\rule{1pt}{\dhatheight}}%
    \smash{\hat{#1}}}}

\newtheorem{theorem}{Theorem}[section]

\newtheorem{proposition}[theorem]{Proposition}

\theoremstyle{definition}
\newtheorem{definition}[theorem]{Definition}

\newtheorem{remark}[theorem]{Remark}

\newcommand{\dd}{\mathrm{d}}

\begin{document}
\title{Gauge potentials on the M5 brane in twisted equivariant cohomotopy}

\author{Pinak Banerjee\thanks{pinakb24@vt.edu}}
\affil{Department of Physics,
\\
Virginia Tech,
\\
850 West Campus Drive,  
\\
Blacksburg,
VA 24061, USA}

\maketitle
\begin{abstract}
In this article, we work out some variations on the discussion of the C-field flux densities in the Sati-Schreiber program. We start by explaining the need for global completion of the field content: the fluxes, the gauge potentials and the gauge transformations on the worldvolume of a single M5 brane in eleven-dimensional supergravity, and how this is encoded by the choice of flux quantization law. 

\smallskip 
Assuming Hypothesis H that the 4-flux in M-Theory is flux quantized in a non-abelian cohomology theory called 4-cohomotopy, and the three-flux on the M5 brane worldvolume in (twisted) 3-cohomotopy, we generalize some previous calculations known in the literature to include twisting by background gravity and placing M5 branes on orbifolds. We show that the null concordances of cohomotopically charged fluxes give rise to the traditional gauge potentials and the null concordances of concordances give rise to the corresponding gauge transformations via surjections, in the cases of tangentially twisted cohomotopy, twistorial cohomotopy and equivariant twistorial cohomotopy. We construct the surjections explicitly for these cases and check the consistency relations with the corresponding Bianchi identities. 

\smallskip 
Thus, we show how the traditional formulas for the local gauge potentials on M5 brane worldvolume  on curved spacetimes and orbifolds are indeed
reproduced by the homotopy theory and, as such, become amenable to global completion in cohomotopical charge quantization.
\end{abstract}
\begin{center}
\begin{minipage}{15.7cm}
\tableofcontents
\end{minipage}
\end{center}
\pagebreak

\section{Introduction}

Traditional foundations of theoretical physics have mostly relied on perturbation theory, applied to weakly coupled systems. There is a general need of understanding of {\it non-perturbative} phenomenon for strongly coupled quantum theories precisely.

\smallskip 
Initially, String Theory originated to explain hadronic bound states, later researchers figured out that it can be a mathematically consistent model for {\it quantum gravity}. Though it admits strongly coupled regimes, dualities in String Theory can relate strong coupling regimes to weak coupling regimes. So, primarily, most of the string dynamics can be understood only perturbatively. 
Now, there emerged {\it seemingly different} five superstring theories, and researchers were looking for a single unified theory to explain them all. Finally, in 1995, all five superstring theories got unified under the name of one single elusive theory with a famous working title: {\it M-Theory }\cite{W-95}, which is the infinity string coupling limit of the Type IIA String Theory. Also see \cite{U} for the motivation by the unification of dualities , and \cite{4} for a good overview.

\smallskip 
M-Theory (see \cite{Duff96}\cite{Duff99}) is an eleven-dimensional strongly coupled theory, reducing all the fields in the five superstring theories to just a pair of fields on the bosonic side: gravity and the three-form C-field. On the fermionic side, we only have the gravitino, the supersymmetric partner of the graviton. Supersymmetry requires us to have the same number of bosonic and fermionic degrees of freedom. The theory is supersymmetric, having 128 on-shell degrees of freedom both on the bosonic (44 for the graviton and 84 for the C field) and the fermionic side (128 for the gravitino). The C-field couples electrically to the M2 brane and magnetically to the M5 brane. The low energy classical limit of M-Theory is believed to be $D=11, \mathcal{N}=1$ SuGra theory {\cite{BST}\cite{MS-05}}{\cite{GSS24-SuGra}}. M-Theory is speculated to be purely non-perturbative by nature, as we do not have any continuous parameter to tune as a coupling constant. So mathematically formulating M-Theory can provide valuable insights to truly understanding non-perturbative QFT and non-perturbative String Theory.

\smallskip 
Non-perturbative effects come in particular not continuously connected to the vacuum, and these are topological configurations whose definition requires to accompany the field's equations of motion by a flux quantization law.  Non-perturbatively, the charge sectors are classified by {\it plain homotopy theory}, the full field content (gauge potentials and corresponding gauge transformations) is encoded in {\it geometric homotopy theory}, and thus they are encoded in {\it non-abelian cohomology theories} \cite[\S 2.1]{FSS23-Char} and their {\it differential} refinement \cite[\S 4.38]{FSS23-Char} thereof. {\it Flux quantization laws} \cite[\S 3]{SS25-Flux} tell us about the concerned (non-abelian) {\it cohomology} theories for the physical fields in a theory; precisely speaking, it imposes the constraint that the deformation classes of the fluxes must be in the image of the character map of the concerned cohomology theory.  Flux quantization laws also encode the {\it extra} field content on top of the topological charges required to make a theory {\it globally} well-defined, which tells us about the consistent {\it gluing} of these higher gauge fields on various patches. Thus, to make sense of the {\it full} theory, we need the flux quantization law.
 Therefore, understanding flux quantization law for a theory really can throw some light on what the theory {\it really} is, and we want to understand it for the case of M Theory \cite{Wi-Flux}.

\smallskip 
The C-field in 11D SuGra is a {\it higher} gauge field, having flux densities as higher-degree differential forms, and with gauge transformations accompanied by higher-order gauge-of-gauge transformations. The general theory of flux quantization says that any cohomology theory flux-quantizing the C-field fluxes has a classifying space whose (minimal) Sullivan model \cite[\S 3.31]{FSS23-Char} has as generators the 
pre-geometric field species subject to differential relations of the same form as the pre-geometric Bianchi identities of the C-field; and {\it rational} homotopy theory \cite[\S 3.2]{FSS23-Char} shows that these are precisely the spaces of the same rational homotopy type of the 4-sphere \cite{Sa-13}. In other words, the equations of motion for the C-field flux in 11D SuGra are the same as the Sullivan model of the 4-sphere in rational homotopy theory; which implies that the concerned cohomology theory should be {\it rationally} equivalent to a non-abelian cohomology theory, namely {\it 4-cohomotopy} \cite[\S 2.10]{FSS23-Char}. In order to lift this to the level of {\it full} cohomology, there are infinitely many inequivalent choices, whose rational avatars coincide with rational 4-cohomotopy. Here, we'll take the simplest choice: 4-cohomotopy {\it itself} to be the {\it full} cohomology theory. This is the idea of {\it Hypothesis H} \cite[\S 2.5]{Sa-13}. It is worthwhile to note that cohomotopy is not like any other cohomology theory, but {\it initial} among all of them. This makes sense as the cohomology theory for M-Theory, which is supposed to be the fundamental foundation of all theories of physics. See \cite{2D}, where a formalism for the integration of non-abelian p-forms is developed, and it is directly related to defining non-abelian homotopies on worldvolume of the M5-brane.

\smallskip 
Now, twisting by background fields is encoded in a variant of these non-abelian cohomology theories, namely {\it twisted} cohomology \cite[\S 2.29]{FSS23-Char}. So here, if we need to consider coupling to gravity, we need {\it twisted} variant of 4-cohomotopy \cite[\S 2.41]{FSS23-Char}, and for taking the gauge potentials into account, we need {\it differential} refinement of 4-cohomotopy \cite[\S 4.45]{FSS23-Char}, where we have the moduli stack of higher gauge fields as the homotopy pullback of the flux densities along the character map. Now if we want to take the effect of orbifolds, we need another flavour of cohomology, namely {\it equivariant cohomology} \cite[\S 2.37]{SS20-Equitwistorial}.
The combined effect of twisting by background fields and orbifolds is encoded in {\it twisted equivariant cohomology} \cite[\S 2.45]{SS20-Equitwistorial}.

\smallskip 
We will  consider objects in M-theory, mainly M5-branes (\cite{Guven}\cite{Wi97}\cite{FKPZ}\cite{PST}, see \cite{Berman} for a review). 
The full definition of classical on-shell field content and dynamics on a single M5 brane worldvolume is quite subtle. 
The structure of the gauge potentials relating to concordances of flux densities and their gauge transformations as higher concordances in 11D SuGra has been worked out recently, and related to traditional formulas, for the case of the plain C-field in 11D supergravity \cite[\S 2.1.4]{GSS24-SuGra}, and for its coupling to the `self-dual' B-field on probe M5-branes \cite[\S 4.1]{GSS24-FluxOnM5}. This further solidifies the role of {\it homotopy theory} in understanding non-perturbative physics.

\smallskip 
In this brief note, we aim to {\it generalize} to spacetimes that are not necessarily flat. Here, we show that these {\it surjections from concordances and higher concordances to the usual (higher) gauge potentials and (higher) gauge transformations exist for a single M5 brane worldvolume even when we take background gravity and effect of orbifolds into account, and thus the local data of (higher) gauge potentials and (higher) gauge transformations are amenable to cohomotopical flux quantization}. To the best of our knowledge, such calculations have not been done in the literature yet. We generalize the calculations for the gauge potentials on the worldvolume of M5 brane in the presence of coupling to background gravitational charges in tangentially twisted cohomotopy \cite[\S 2.41]{FSS23-Char} and then we generalize to twistorial cohomotopy \cite[\S 2.45]{FSS23-Char}, which encodes information about $\frac{1}{2}$-BPS immersions of M5 brane solutions in 11D SuGra with a Chern-Simons gauge field, or alternatively about heterotic M-Theory. Lastly, we take into account the effect of orbifolds alongside background gravity on M5 brane worldvolume, and this is encoded in twistorial equivariant cohomotopy \cite[\S 2.48]{SS20-Equitwistorial}.

\smallskip 
We use the language of rational homotopy theory \cite[\S 3.2]{FSS23-Char} and differential forms. Thus, for the first two cases of tangentially twisted and twistorial cohomotopy, we'll get to see background gravitational charges as well as gauge potentials arising from concordances and their gauge transformations from concordances of concordances. Here, we'll work in the setup of {\it twisted differential cohomotopy} \cite[\S 5.27]{FSS23-Char}, which is a combined generalization of twisted and differential variants of cohomology at the rational level. When we take the effect of orbifolds into account alongside gravitational charges and gauge potentials in twistorial equivariant cohomotopy, we need to work in the setup of {\it twisted equivariant differential cohomotopy}.

\smallskip 
We start with the Bianchi identities, modified by the gravitational contributions and the fixed locus of the orbifold action. Next, we show here that there exist surjections from the concordances to the usual gauge potentials and from the concordances of concordances to the usual gauge transformations. We explicitly spit out the surjections, and also check their consistency with the (twisted) Bianchi identities.

\smallskip

The paper is organized as follows. In section 2, we'll briefly review flux quantization law \cite[\S 3]{SS25-Flux}. In section 3, we generalize the computations in \cite[\S 4.1]{GSS24-FluxOnM5} to 
include the tangential twists (following \cite{FSS20-H}) which couple these fluxes (on M5-probes of 11D SuGra) to gravity, shifting the Bianchi identities by Pontrjagin forms as known from the Green-Schwarz mechanism (for which cf. review in \cite[\S 12]{FSS23-Char}). Then, we go one step further and calculate them for the case of twistorial cohomotopy \cite[\S 2.14]{FSS22-GS}, which is closely related to heterotic M-Theory. In section 4, we put the M5 brane on orbifolds, along with background gravity, and repeat the calculations for the case of $\mathbb{Z}_{2}$-equivariant $\mathrm{Sp}(1)$-parametrized twistorial cohomotopy.

\bigskip 
\noindent {\bf Acknowledgements.} The author would like to thank Hisham Sati at New York University Abu Dhabi for useful discussions and helping a lot with the manuscript.

\newpage 
\section{Flux Quantization and Hypothesis H}

In this section, we will review some aspects of flux quantization \cite[\S 3]{SS25-Flux} and specialize to Hypothesis H \cite[\S 2.5]{Sa-13}.
Given a suitable higher gauge theory, then (as reviewed in \cite{SS25-Flux}): 
\begin{itemize}
\item[--]its Bianchi identities characterize the duality-symmetric flux densities jointly as a closed (flat) differential form with coefficients in a characteristic $L_\infty$-algebra $\mathfrak{a}$, 
\item[--] the admissible form of the corresponding gauge potentials and their (higher) gauge transformations may systematically be deduced as (iterative) deformations (concordances) of such closed $\mathfrak{a}$-valued differential forms, 
\item[--] in a way that prepares the theory for global completion by  {\it flux quantization}.
\end{itemize}

One aspect that has been largely overlooked in the existing literature is that the Bianchi identities for these (higher) gauge fields hold {\it locally}. Thus, the local gauge potential form fields {\it do not} represent the entire configuration of the higher gauge fields. Further field content or transition data is essential to {\it glue} these gauge potentials on higher intersections of the form $\tilde{X}^{\times^{n}_{X}}$, where X is the spacetime, $\tilde{X}$ is the surjective submersion and we have the projector map $p_{X}:\tilde{X}\rightarrow{X}$.

The admissible choices of this remaining field content are captured by {\it flux quantization laws}. To obtain a comprehensive understanding of flux quantization, information regarding transition data and submersions is essential. This {\it extra} choice actually completes the global picture of the (higher) gauge theories.

For example, consider the Maxwell theory with Bianchi identity (in the absence of currents) for $F_{2}\in \Omega^{2}_{\mathrm{dR}}(X)$
\begin{equation}
    \dd F_{2}=0.
\end{equation}
These identities do not need a transition data to be stated. But in flux quantization, when we declare that $F_{2}$ is the curvature of a $U(1)$ principal bundle with connection $A_{1}$ and thus the theory is quantized in ordinary integral cohomology in degree two denoted by $H^{2}(X, \mathbb{Z})$, then $A_{1}$ is no longer a 1-form on the spacetime X, but on the submersion $\tilde{X}$
\begin{equation}
    \begin{array}{l}
        A_{1}\in \Omega_{\mathrm{dR}}^{1}(\tilde{X}),\qquad \dd A_{1}=F_{2}\,.
\end{array}
\end{equation}
Now, as mentioned above, these local gauge potentials are {\it not} enough to capture the global aspects of the theory, but they do require some extra field content to encode their gluing on higher intersections.
Thus, we have here
\begin{equation}
\renewcommand{\arraystretch}{1.3}
\begin{array}{l}
    A_{0}\in \Omega_{\mathrm{dR}}^{0}(\tilde{X}^{\times^{2}_{X}}),\quad
    \dd A_{0}= \mathrm{pr}_{2}^{*}A_{1}-\mathrm{pr}_{1}^{*}A_{1}\\ \mathrm{pr}_{12}^{*}A_{0}+\mathrm{pr}_{23}^{*}A_{0}-\mathrm{pr}_{13}^{*}A_{0} = n_{123}, \hspace{0.4cm}\text{where} \hspace{0.1cm}n_{123}\in\hspace{0.2cm}C^{0}(\tilde{X}^{\times^{3}_{X}}, \mathbb{Z}).\\
    n_{234}- n_{134}+ n_{124}-n_{123}=0 \hspace{0.2cm} \text{on}\hspace{0.1cm}\tilde{X}^{\times^{4}_{X}}.
    \end{array}
\end{equation}
The cocycle condition on triple intersections $\tilde{X}^{\times^{3}_{X}}$ involving the projectors $\mathrm{p}_{ij}:\tilde{X}^{\times^{3}_{X}}\rightarrow{\tilde{X}^{(i)}\times_{X}\tilde{X}^{(j)}}$ enforces the Dirac quantization condition.

Thus, this extra field content $A_{0}$ encodes the $\check{\mathrm C}$ech cocycle data required for making $A_{1}$ to be a {\it globally} defined connection on a U(1) principal bundle over the spacetime X.

Ordinary flux quantization in $\check{\mathrm C}$ech cohomology has been discussed in \cite[\S 2]{A-85}.

As a second motivating example, consider the $H_{3}$ flux in String Theory. It admits the Bianchi identity\begin{equation}
    \dd H_{3}=0.
\end{equation}
Now, when we declare it to be flux quantized in ordinary integral cohomology in degree three denoted by $H^{3}(X, \mathbb{Z})$, and view $H_{3}$ as the curvature of a circle 2-bundle or a $U(1)$ gerbe, then we have a 2-connection on this circle gerbe, namely $B_{2}$, also known as the NS-NS sector Kalb Ramond field. This $B_{2}$ is no longer a differential two-form on the spacetime X, but on $\tilde{X}$  satisfying the Bianchi identity
\begin{equation}
    B_{2}\in \Omega_{\mathrm{dR}}^{2}(\tilde{X}),\qquad \dd B_{2}=H_{3}.
\end{equation}
Now, this $B_{2}$ alone is not enough to capture the global aspects of the full theory, but we also need to have some extra field content which captures the gluing of the $B_{2}$ on higher intersections. Then we have 
\begin{equation}
\renewcommand{\arraystretch}{1.3}
\begin{array}{l}
    B_{1}\in \Omega_{\mathrm{dR}}^{1}(\tilde{X}^{\times^{2}_{X}}), \quad 
    \dd(B_{1})_{12}= \mathrm{pr}_{2}^{*}B_{2}-\mathrm{pr}_{1}^{*}B_{2}
    \\ 
    B_{0}\in \Omega_{\mathrm{dR}}^{0}(\tilde{X}^{\times^{3}_{X}}), \quad 
    \dd(B_{0})_{123}= \mathrm{pr}_{23}^{*}B_{1}-\mathrm{pr}_{13}^{*}B_{1} + \mathrm{pr}_{12}^{*}B_{1}\\ \mathrm{pr}_{123}^{*}B_{0}+\mathrm{pr}_{134}^{*}B_{0}-\mathrm{pr}_{124}^{*}B_{0}- \mathrm{pr}_{234}^{*}B_{0} = m_{1234}\hspace{1cm}  \text{where}\hspace{0.2cm}m_{1234}\in C^{0}(\tilde{X}^{\times^{4}_{X}}, \mathbb{Z}) .\\
    m_{2345}-m_{1345}+m_{1245}-m_{1235}+m_{1234}=0 \hspace{1cm}\text{on}\hspace{0.2cm}\tilde{X}^{\times^{5}_{X}}\,.
    \end{array}
\end{equation}
This $B_{1}$ is {\it still} insufficient to complete the full theory, as we need information about gluing of $B_{1}$ on further higher intersections which is encoded in $B_{0}$.

The higher ($\check{\mathrm C}$ech) cocycle condition on quartic intersections $\tilde{X}^{\times^{4}_{X}}$ involving the projectors $\mathrm{p}_{ijk}:\tilde{X}^{\times^{4}_{X}}\rightarrow{\tilde{X}^{(i)}\times_{X}\tilde{X}^{(j)}\times_{X}\tilde{X}^{(k)}}$ encodes the integral quantization condition in degree three.

Thus, these extra field content $B_{1}, B_{0}$ (along with the (higher) cocycle conditions) encode the data required for making $B_{2}$ to be a {\it globally} defined 2-connection on a U(1) gerbe over the spacetime X.

As a final motivating example, consider the $G_{4}$ flux in 11D supergravity. It admits the Bianchi identity
\begin{equation}
    \dd G_{4}=0.
\end{equation}
Now, when we declare it to be flux quantized in ordinary integral cohomology in degree four denoted by $H^{4}(X, \mathbb{Z})$, and see $G_{4}$ as the curvature of a circle 3-bundle or a $U(1)$ 2-gerbe, then we have a 3-connection on the circle 2-gerbe, namely the $C_{3}$ field. This $C_{3}$ is not a differential three-form on the spacetime X, but on $\tilde{X}$ satisfying the Bianchi identity
\begin{equation}
    C_{3}\in \Omega_{\mathrm{dR}}^{3}(\tilde{X}),\qquad \dd C_{3}=G_{4}.
\end{equation}
Now, this $C_{3}$ alone is not enough to encode the global aspects of the full theory, but we also require to have some extra field content which captures the gluing of the $C_{3}$ on higher intersections. We thus have
\begin{equation}
\renewcommand{\arraystretch}{1.3}
\begin{array}{l}
    C_{2}\in \Omega_{\mathrm{dR}}^{2}(\tilde{X}^{\times^{2}_{X}}), \quad 
    \dd(C_{2})_{12}= \mathrm{pr}_{2}^{*}C_{3}-\mathrm{pr}_{1}^{*}C_{3}\\ 
    C_{1}\in \Omega_{\mathrm{dR}}^{1}(\tilde{X}^{\times^{3}_{X}}), \quad 
    \dd(C_{1})_{123}= \mathrm{pr}_{23}^{*}C_{2}-\mathrm{pr}_{13}^{*}C_{2} + \mathrm{pr}_{12}^{*}C_{2}\\ 
    C_{0}\in \Omega_{\mathrm{dR}}^{0}(\tilde{X}^{\times^{4}_{X}}), \quad
    \dd(C_{0})_{1234}= \mathrm{pr}_{234}^{*}C_{1}-\mathrm{pr}_{134}^{*}C_{1} + \mathrm{pr}_{124}^{*}C_{1}-\mathrm{pr}_{123}^{*}C_{1}\\ \mathrm{pr}_{2345}^{*}C_{0}+\mathrm{pr}_{1245}^{*}C_{0}+\mathrm{pr}_{1234}^{*}C_{0}-\mathrm{pr}_{1345}^{*}C_{0}- \mathrm{pr}_{1235}^{*}C_{0}= l_{12345} \hspace{1cm} \text{where}\hspace{0.2cm}l_{12345}\in C^{0}(\tilde{X}^{\times^{5}_{X}}, \mathbb{Z}) \\
    l_{23456}-l_{13456}+l_{12456}-l_{12356}+l_{12346}-l_{12345}=0 \hspace{1cm}\text{on}\hspace{0.2cm}\tilde{X}^{\times^{6}_{X}}\,.
    \end{array}
\end{equation}
Now, this $C_{2}$ is {\it yet} insufficient to encode information about the full SuGra theory, as we still need information about the gluing of $C_{2}$ on higher intersections, which is encoded by $C_{1}$, gluing of which on further higher intersections is encoded by $C_{0}$.

The higher ($\check{\mathrm C}$ech) cocycle condition on quintic intersections $\tilde{X}^{\times^{5}_{X}}$ involving the projectors $\tilde{X}^{\times^{5}_{X}}\rightarrow{\tilde{X}^{(i)}\times_{X}\tilde{X}^{(j)}\times_{X}\tilde{X}^{(k)}\times_{X}\tilde{X}^{(l)}}$ encodes information about the integral $G_{4}$ flux quantization condition in degree four.

Thus, these {\it extra} field content $C_{2}, C_{1}, C_{0}$ (along with the (higher) cocycle conditions) encode the data required for making $C_{3}$ gauge potential in 11D SuGra to be a {\it globally} defined 3-connection on a U(1) 2-gerbe over the spacetime X.

One might think of choosing ordinary integral cohomology in degree four $H^{4}(X, \mathbb{Z})$ as the possible flux quantization law for 11D SuGra, but this fails to capture the non-linear Bianchi identity 
$$\dd 2G_{7}=G_{4}\wedge G_{4}$$
which is a hint that the probable flux quantization law for 11D SuGra must be a {\it non-abelian} cohomology theory.
Thus, {\it flux quantization} is essential for understanding the {\it full} (higher) gauge theories.

The flux quantization in non-abelian A-cohomology \cite[\S 2.1]{FSS23-Char}
\begin{equation}
    A(X):= \pi_{0}\mathrm{Maps}(X,A)
\end{equation} implies that the Bianchi identities for the flux densities should lie in the image of the non-abelian character map for non-abelian A-cohomology \cite[\S 4.3]{FSS23-Char}
\begin{equation}
    \mathrm{ch}_{A}:A(X)\longrightarrow H^{1}_{\mathrm{dR}}(X; \mathfrak{l}A).
\end{equation}
where $\mathfrak{l}A\cong{\mathfrak{a}}\in \mathrm{L_{\infty}Algs}$. Therefore, flux densities are classes in $\mathfrak{a}$-valued deRham cohomology.

Given a choice of flux quantization law A, the moduli stack of higher gauge fields $\hat{\mathrm{A}}$ with A-quantized flux densities classifies the {\it differential} non abelian A-cohomology $\hat{A}(X)$ \cite[\S 4.38]{FSS23-Char}, where:
\begin{equation}
    \hat{\mathrm{A}}(X):= \pi_{0}\mathrm{Maps}(X, \hat{A}) := {\mathrm{A}(X)}\times_{H^{1}_{\mathrm{dR}}(X; \mathfrak{l}A)} \Omega^{1}_{\mathrm{dR}}(X; \mathfrak{l}A)
\end{equation}
and the higher gauge potentials $\vec{A}$ are viewed as {\it homotopies} between the character map of charges $\chi$ and the flux densities $\vec{F}$
\begin{equation}
    \vec{A}: \mathrm{ch}(\chi) \implies \vec{F}\,.
\end{equation}
For more details about non-abelian flux quantization, see \cite[\S 3]{SS25-Flux}.

For 11D SuGra, the conjectured flux quantization law is {\it Hypothesis H}
\cite[\S 2.5]{Sa-13}, which states that the $G_{4}$ flux is quantized in a non-abelian cohomology theory called {\it 4-cohomotopy} 
\begin{equation}
    \pi^{4}(X):= \pi_{0}\mathrm{Maps}(X, S^{4})
\end{equation} and the `self dual' $H_{3}$ flux on the worldvolume of M5 brane $\Sigma$ is charge quantized in {\it 3-cohomotopy} (twisted by background M-brane charges $[c_{3}, c_{6}]\in \pi^{4}(X)$)
\begin{equation}
    \pi^{3+\phi^{*}(c_{3},c_{6})}(\Sigma):= \pi_{0}\mathrm{Maps}(\Sigma, S^{7})_{/S^{4}}
\end{equation} where we have the embedding map $\phi:\Sigma\longrightarrow X$.

Upon a secondary completion to {\it differential 4-cohomotopy} $\hat{\pi}^{4}(X)$ we recover the $C_{3}$ three-form gauge potential in 11D SuGra and from $\hat{\pi}^{3+\phi^{*}(c_{3},c_{6})}(\Sigma)$ we obtain the $B_{2}$ two-form gauge potential on the worldvolume of M5 brane.

In this paper, we will explore these potentials from concordances and their gauge transformations from higher concordances for the case of twisting by background gravity, which modifies the Bianchi identities by Pontrjagin forms.

Say, for the $G_{7}$ flux, we'll now have 
$$\dd 2G_{7}=G_{4}\wedge G_{4}- \tfrac{1}{16}p_{1}(\omega)\wedge{p_{1}(\omega)}- \chi_{8}$$
and also for the $H_{3}$ flux on M5 brane worldvolume we have
$$\dd H_{3}=\tilde{G_{4}}-\tfrac{1}{2}p_{1}(\omega).$$ We will return to them in the next section, along with the corresponding gauge potentials and {\it} modified Bianchi identities.

\section{Gauge potentials on M5-probes of 11D SuGra in Twisted cohomotopy}
We want to couple the theory to background gravitational charges. 

\subsection{In Tangentially twisted cohomotopy}
\medskip
Since the only other bosonic field apart from $C_{3}$ in 11D SuGra is the metric, the concerned cohomology theory must be the {\it tangentially} twisted 4-cohomotopy. In other words, C-field flux is charge quantized in tangentially Sp(2) twisted 4-cohomotopy \cite[\S 3.1]{FSS20-H}.\footnote{D=11 SuGra stands out as its only other bosonic field besides gravity is the $C$-field. This means that the possible twistings of the C-field flux quantization can only be by the gravitational field, namely by the Spin-frame bundle of the spacetime $X$ (the principal bundle underlying its tangent bundle). By the general rules of twisted cohomology and assuming Hypothesis H on flat spacetimes, this implies that the possible twistings are given by $\infty$-actions of (subgroups of) the Spin-4 group on the 4-sphere. The canonical action is that of Spin(5) via the defining action of SO(5) on $S^{4} = S(\mathbb{R}^{5})$ regarded as the unit sphere in $\mathbb{R}^{5}$. There is an isomorphic (but subtly different, which we won't mention here) action of $Sp(2)\simeq Spin(5)$ on $S^{4}$. The integral cohomology of $S^{4}\sslash Sp(2)\simeq \mathrm{BSpin}(4)$ is generated from $\frac{1}{2}p_{1}$ and the combination $\frac{1}{2}\chi_{4}+\frac{1}{4}p_{1}$.}
We define {\it tangentially twisted or J twisted Cohomotopy} (described in \cite[\S 2.1]{FSS20-H}, \cite[\S 2.41]{FSS23-Char})
\begin{equation}
    \pi^{4+\tau}(X):=\mathrm{H}^{\tau}(X; S^{4})
    \end{equation}
    where we have the {\it twist} $\tau:X\longrightarrow{\mathrm{BSp}}(2)$.\footnote{The first non-trivial check of the tangential twisting is its implication of the notorious shifted integral flux quantization of the 4-flux density, a widely expected hallmark of M-theory. It says that not the deRham cohomology class of $G_{4}$ but its shift by one-fourth of the Pontrjagin 4-form on spacetime is the real image of an integral cohomology class. In non-abelian cohomology this condition falls out naturally as the pullback of $\frac{1}{2}\chi_{4}$ (half of the Euler class or the volume form on $S^{4}$) is interpreted as the $G_{4}$-flux under Hypothesis H, and the pullback of  $\frac{1}{2}\chi_{4}+\frac{1}{4}p_{1}$ from $S^{4}\sslash Sp(2)$ to $X$ implies $G_{4}+\frac{1}{4}p_{1}$ is integral itself as it is the image of the pullback of an integral form.}

    Also, we have the Borel-equivariantized quaternionic Hopf fibration $$h_{\mathbb{H}}\sslash\mathrm{Sp}(2):S^{7}\sslash\mathrm{Sp}(2)\longrightarrow{S^{4}}\sslash\mathrm{Sp}(2)$$ which will be useful here.
\begin{definition}[\bf Flux densities]
We are concerned with {\it flux densities} of the following form (\cite[\S 3.7]{FSS20-H}): 
\begin{equation}
\renewcommand{\arraystretch}{1.3}
  \label{TheFluxDensities}
	\Omega_{\mathrm{dR}}^{1}\left(
      -;
      \,
      \mathfrak{l}_{S^{4}\sslash \mathrm{Sp}{(2)}} S^{7}\sslash \mathrm{Sp}(2)\right)_{\mathrm{clsd}}
        :=
      \left\{\begin{array}{l|l}
		H_{3} & 
          \mathrm{d}\, H_{3}=\tilde{G}_{4}-\tfrac{1}{2} p_{1}(\omega) \\
		G_{7} & \mathrm{d}\, G_{7}=\tfrac{1}{2} \tilde{G}_{4}\left(\tilde{G}_{4}-\tfrac{1}{2} p_{1}(\omega)\right) \\
		G_{4} & \mathrm{d}\, G_{4}=0 \\
		\tfrac{1}{2} p_{1}(\omega) & d\left(\tfrac{1}{2} p_{1}(\omega)\right)=0
	\end{array}\right\}
    \,,
\end{equation}
where we set
\begin{equation}
    \tilde{G}_{4}
    := 
    G_{4} + \tfrac{1}{4} p_{1}(\omega)
\end{equation}
and where $p_1(\omega)$ denotes the first Pontrjagin form for a given spin-connection on spacetime:
\begin{equation}
    \omega 
    \in 
    \Omega^{1}_{dR}\big(
      \tilde{X}; 
      \mathfrak{so}(d,1)
    \big)
    \,,
\end{equation}
 where $\tilde{X}:= \sqcup_{i} U_{i}$ for $\big\{ U_{i} \xrightarrow{\iota_i} X \}_{i \in I}$ an open cover of spacetime.
\end{definition}
\begin{remark}
Here we have already assumed $\widehat{M5}$ structure\footnote{$\widehat{M5}$ is the M theoretic analog of the Fivebrane structure for NS5 brane sigma-model, a topological condition on spacetime that is a higher degree analog of ``String structure''. For our case, $\widehat{M5}$ denotes the trivialization of the Euler class, a degree-8 polynomial in the Pontrjagin forms of spacetime to vanish. It implies the (half) integrality of the M2-brane Page charge. }, thus the eight form $\chi_{8}$ coming from $l B\mathrm{Sp}(2)$ vanishes, cf. \cite[Ex. 3.2]{FSS21-Hopf}. Also note that
\begin{equation}
    \mathrm{CE}\big(
      \mathfrak{l} B\mathrm{Sp}(2)
    \big)
      \;\simeq\;
    \mathbb{R}\left[
      \begin{array}{c}
        \tfrac{1}{2}p_{1}(\omega)
        \\
        \chi_{8}
      \end{array}
    \right]
      \Big/
    \left(
      \begin{array}{ccl}
        \mathrm{d} (\tfrac{1}{2}p_{1}(\omega)) &=& 0, 
        \\
        \mathrm{d}\chi_{8} &=& 0
      \end{array}
    \right).
\end{equation}
\end{remark}
\begin{remark}
    Let $\omega$ be a connection 1-form on a principal bundle with structure group $SO(d,1)$ taking values in the Lie algebra $\mathfrak{so}(d,1)$, that is 
$\omega\in \Omega_{\mathrm{dR}}(\tilde{X}, \mathfrak{so}(d,1))$.
Then the Chern-Simons 3-form of the $\mathfrak{so}(d,1)$ valued spin connection $\omega$, on each $U_i$ is given by,
\begin{equation}
    \mathrm{CS}(\omega)
    \,=\, {\mathrm Tr}(\omega\wedge d\omega + \tfrac{2}{3} \omega \wedge \omega \wedge \omega)
\end{equation}
where the trace is taken over the fundamental/vector representation of $\mathfrak{so}(d,1)$,
satisfying 
\begin{equation}
  \mathrm{d} \, \mathrm{CS}(\omega) 
  \,=\, 
  \iota_i^\ast \tfrac{1}{2} p_{1}(\omega)
  \;\;\;
  \in
  \;
  \Omega^4_{\mathrm{dR}}(U_i)
  \,.
\end{equation}
The Chern–Simons form itself is not gauge invariant, but its exterior derivative is. 

{\it Globally}, it is a 3-form supported on $\tilde{X}$, that is, $$CS(\omega)\in \Omega_{\mathrm{dR}}^{3}(\tilde{X}, \mathfrak{so}(d,1).$$

\begin{definition}[\bf Traditional gauge potentials and gauge transformations]
\label{TraditionalFormulas}
Given such flux densities \eqref{TheFluxDensities} on some $U_i$, 
a traditional choice of {\it gauge potentials} is 
\begin{equation}
  \label{TraditionalGaugePotentials}
    \renewcommand{\arraystretch}{1.4}
	\left\{\begin{array}{l|l}
		C_{3} \in \Omega_{dR}^{3}\left(U_{i}\right) & \dd C_{3}=\tilde{G}_{4} \\
		C_{6} \in \Omega_{dR}^{6}\left(U_{i}\right) & \dd C_{6}= G_{7}-\tfrac{1}{2} C_{3}\left(\tilde{G}_{4}-\tfrac{1}{2} p_{1}(\omega)\right) \\
		B_{2} \in \Omega_{dR}^{2}\left(U_{i}\right) & \dd B_{2}= H_{3}-C_{3}+ \mathrm{CS}(\omega)
	\end{array}\right\}.
\end{equation}
Now, a {\it gauge transformation} between a pair of such gauge potentials is
\begin{align}
  \label{GaugeTransformations}
  \hspace{-3mm}
	\left(C_{3}, C_{6}, B_{2} \right) \sim\left(C_{3}^{\prime}, C_{6}^{\prime}, B_{2}^{\prime} \right) & \Longleftrightarrow 
    \nonumber 
    \\
    \exists & \left\{\!\!
    \renewcommand{\arraystretch}{1.3}
    \begin{array}{l|l}
		C_{2} \in \Omega_{dR}^{2}\left(U_{i}\right) & \dd C_{2}=C_{3}^{\prime}-C_{3} \\
		C_{5} \in \Omega_{dR}^{5}\left(U_{i}\right) & \dd C_{5}=C_{6}^{\prime}-C_{6}-\tfrac{1}{2} C_{3}^{\prime} C_{3}
        -
       \tfrac{1}{4} C_{2}
        \,
        p_1(\omega)
        \\
		B_{1} \in \Omega_{dR}^{1}\left(U_{i}\right) & \dd B_{1}=B_{2}^{\prime}-B_{2}+C_{2} -\int_{t\in[0,1]} \mathrm{Tr}((\omega^{\prime}-\omega)\wedge\omega
        _{t}) dt
        \\
	\end{array}
    \!\!
    \right\}
\end{align}
where $$
\omega_{t}=\omega+t(\omega^{\prime}-\omega)\,.
$$
\end{definition}

Our goal now is to (re-)derive these traditional formulas (Def. \ref{TraditionalFormulas}) and to lift them to {\it concordances} of flat $L_\infty$-valued forms. To recall
\begin{definition}[\bf Null-concordances of flux densities]
\label{NullConcordancesOfFluxDensities}
Given flux densities as in \eqref{TheFluxDensities}, we say that a {\it null concordance} is
\begin{equation}
  \label{NullConcordances}
	\begin{array}{c}
		\left(\hat{G}_{4}, \hat{G}_{7}, \hat{H}_{3}, \hat{p}_{1}(\omega)\right) \in \Omega_{\text{dR}}^{1}\left(U_{i} \times[0,1], \mathfrak{l}_{S^{4}\sslash \mathrm{Sp}(2)} S^{7}\sslash \mathrm{Sp}(2)\right)_{ \mathrm{clsd} }
	\end{array}
\end{equation}
    such that
\begin{equation}
  \renewcommand{\arraystretch}{1.7}
    \begin{array}{l}
		\iota_{0}^{*}\left(\hat{G}_{4}, \hat{G}_{7}, \hat{H}_{3}, \hat{p}_{1}(\omega)\right)
          =
        0 
        \\
		\iota_{1}^{*}\left(\hat{G}_{4}, \hat{G}_{7}, \hat{H}_{3}, \hat{p}_{1}(\omega)\right)
          =
        \left(G_{4}, G_{7}, H_{3}, p_{1}(\omega)\right)\,;
    \end{array}
\end{equation}
and given a pair of these, we say that a {\it concordance-of-concordances} between them is
\begin{equation}
  \label{ConcordanceOfConcordances}
	\begin{array}{l}
		({\hat{\hat{G}}}_{4}, {\hat{\hat{G}}}_{7}, {\hat{\hat{H}}}_{3}, {\hat{\hat{p}}}_{1}(\omega, \omega^{\prime})) 
              \;\in\; 
            \Omega_{dR}^{1}(U_{i} \times \underbrace{[0,1]}_{t} \times \underbrace{[0,1]}_{s}; \mathfrak{l}_{S^{4}\sslash \mathrm{Sp}(2)}S^{7}\sslash \mathrm{Sp}(2))_{ \mathrm{clsd }}
    \end{array}
\end{equation}
such that
\begin{equation}
  \renewcommand{\arraystretch}{1.8}
  \begin{array}{l}
	\iota_{s=0}^{*}\left({\hat{\hat{G}}}_{4}, {\hat{\hat{G}}}_{7}, {\hat{\hat{H}}}_{3}, {\hat{\hat{p}}}_{1}(\omega, \omega^{\prime}) \right)
    =
    \left(\hat{G}_{4}, \hat{G}_{7}, \hat{H}_{3}, \hat{p}_{1}(\omega)\right) 
    \\
	\iota^{*}_{s=1}\left({\hat{\hat{G}}}_{4}, {\hat{\hat{G}}}_{7}, {\hat{\hat{H}}}_{3}, {\hat{\hat{p}}}_{1}(\omega, \omega^{\prime})\right)
    =
    \left(\hat{G}_{4}^{\prime}, \hat{G}_{7}^{\prime}, \hat{H}_{3}^{\prime}, \hat{p}_{1}(\omega^{\prime})\right) \\
			\iota^{*}_{t=0}\left({\hat{\hat{G}}}_{4}, {\hat{\hat{G}}}_{7}, {\hat{\hat{H}}}_{3}, {\hat{\hat{p}}}_{1}(\omega, \omega^{\prime}) \right)=0 \\
			{ \iota}^{*}_{t=1}\left({\hat{\hat{G}}}_{4}, {\hat{\hat{G}}}_{7}, {\hat{\hat{H}}}_{3}, {\hat{\hat{p}}}_{1}(\omega, \omega^{\prime}) \right) = \operatorname{pr}_{U_{i}}^{*}\left(G_{4}, G_{7}, H_{3}, p_{1}(\omega)\right),
		\end{array}\\
    \end{equation}
	where 
    $$
      \mathrm{pr}_{U_{i}}: U_{i} \times[0,1]_{s} \longrightarrow U_{i}
    $$
    is projection onto the first factor.
\end{definition}    
\begin{remark}
    In the remainder of this article, we can identify $2 \mathrm{CS}(\omega)$ with $s_{3}(\omega)$, which will appear later in the expression for $\hat{p}_{1}(\omega)$. 
Also, $s_{2}(\omega, \omega')$ denotes the concordance between two Chern-Simons forms with spin connections $\omega$ and $\omega^{\prime}$, the usual trangression form  $s_{2}(\omega, \omega')$ given by 
\begin{equation}
\dd s_{2}(\omega, \omega'):= s_{3}^{\prime}- s_{3}= 2[\mathrm{CS}(\omega^{\prime})- \mathrm{CS}(\omega)]\\
=2 \int_{t\in[0,1]}\frac{d}{dt}\mathrm{CS}(\omega
_{t})dt
\end{equation}
where $\omega_{t}=\omega+ t (\omega^{\prime}-\omega)$.
Explicitly, it is given in terms of the spin connections $\omega$, $\omega^{\prime}$ as
\begin{equation}
s_{2}(\omega, \omega')=2 \int_{t\in[0,1]} \mathrm{Tr}(\omega^{\prime}-\omega)\wedge \omega_{t} dt.
\end{equation}
The form $s_{2}(\omega, \omega')$ will appear later in the expression for $\doublehat{p}_{1}(\omega, \omega')$.
\end{remark}
\begin{remark}
     In the remainder of the article, we will write  $s_{3}(\omega), s_{3}(\omega'), p_{1}(\omega), s_{2}(\omega, \omega'), \hat{p}_{1}(\omega), \hat{p}_{1}(\omega^{\prime})$, $\doublehat{p}_{1}(\omega, \omega^{\prime})$ as 
     $s_{3},s_{3}^{\prime}, p_{1}, s_{2}, \hat{p}_{1}, \hat{p}_{1}^{\prime}, \doublehat{p}_{1}$, 
     respectively and omit the arguments involving $\omega,\omega^{\prime}$ for notational brevity.

\end{remark}
\end{remark}
\begin{proposition}[\bf Gauge potentials from null concordances]:
Given null concordances \eqref{NullConcordances} the following formulas define the gauge potentials for the flux densities \eqref{TheFluxDensities}
\vspace{-2mm} 
\begin{equation}
  \label{FluxDensitiesFromNullConcordances}
  \renewcommand{\arraystretch}{1.7}		
		\begin{array}{rcl}
			2 \mathrm{CS}(\omega)
              &:=& 
            s_{3}=\int_{[0,1]}\hat{p}_{1}
            \\
			C_{3}
              &:=&
            \int_{[0,1]} \hat{G}_{4}+\tfrac{1}{4} \int_{[0,1]} \hat{p}_{1}=\int_{[0,1]} \hat{\tilde{G}}_{4}
            \\
			C_{6}
              &:=&
            \int_{[0,1]} (\hat{G}_{7}-\tfrac{1}{2}\left(\int_{[0,-]}\hat{\tilde{G}}_{4}\right)\left(\hat{\tilde{G}}_{4}-\tfrac{1}{2} \hat{p}_{1}\right))
            \\
			B_{2} &:=& \int_{[0,1]} \hat{H}_{3}
            \,.
		\end{array}
\end{equation}
\end{proposition}
Moreover, this construction is surjective, in that every traditional gauge potential arises from some null-concordance this way.
\begin{proof}
First to see that the claimed flux densities satisfy their defining differential equations, we compute as follows:
\begin{equation}
\renewcommand{\arraystretch}{1.5}
    \begin{array}{ll}
   \dd(2CS(\omega)) &=\dd s_{3}
       =\int_{[0,1]}\hat{{p_{1}}}
      \\
&= \iota^{*} _{1}\hat{{p_{1}}}-\iota^{*}_{0} \hat{{p_{1}}} -\int_{[0,1]} d \hat{{p_{1}}}
      \\
 &      =p_{1}-0-0
      \\[-2pt]
 &   =p_{1}\,.
  \end{array}
\end{equation}
Moving onto $C_{3}$ we find,
\begin{equation}
 \renewcommand{\arraystretch}{1.5}
  \begin{array}{rcl}
	\dd C_{3}
      &=& \dd \int_{[0,1]}\hat{\tilde{G_{4}}}
      \\
      &=&
      \iota^{*} _{1}\hat{\tilde{G_{4}}}-\iota^{*}_{0} \hat{\tilde{G_{4}}} -\int_{[0,1]} d \hat{\tilde{G_{4}}}
      \\
      &=& \tilde{G_{4}}-0-0
      \\[-2pt]
      &=&
      \tilde{G_{4}}\,.
  \end{array}
\end{equation}
Next, we check the expression for $C_{6}$.
\begin{equation}
   \renewcommand{\arraystretch}{1.6}
	\begin{array}{rcl}
		\dd C_{6}
        &=&
        \dd\left(\int_{[0,1]} \hat{G}_{7}-\tfrac{1}{2} \int_{[0,-]} \hat{\tilde{G}}_{4}\left(\hat{\tilde{G}}_{4}-\tfrac{1}{2} \hat{p}_{1}\right)\right)
        \\
        &=& \iota^{*}_{1}\left(\hat{G}_{7}-\tfrac{1}{2} \hat{C}_{3}\left(\hat{\tilde{G}}_{4}-\tfrac{1}{2} \hat{p}_{1}\right)\right)-0-\int_{[0,1]} d\left[\hat{G}_{7}-\tfrac{1}{2} \hat{C}_{3}\left(\hat{\tilde{G}}_{4}-\tfrac{1}{2} \hat{p}_{1}\right)\right] 
        \\
		&=&
        G_{7}-\tfrac{1}{2}{C_{3}}{}\left(\tilde{G}_{4}-\tfrac{1}{2}p_{1}\right)-0-0= G_{7}-\tfrac{1}{2}{C_{3}}{}\left(\tilde{G}_{4}-\tfrac{1}{2}p_{1}\right).
	\end{array}
\end{equation}
Next, we would like to check the equation for $B_{2}$.
\begin{equation}
  \renewcommand{\arraystretch}{1.7}
	\begin{array}{rcl}
		\dd B_{2}
        &=&
        \dd \int_{[0,1]} \hat{H}_{3}=\iota^{*}_{1} \hat{H_{3}}-\iota^{*}_{0}\hat{{H}_{3}}-\int_{[0,1]} \dd \hat{H}_{3}
        \\
        &=&
        H_{3}-0-\int_{[0,1]}\left(\hat{G}_{4}-\tfrac{1}{4} \hat{p}_{1}\right)
        \\
        &=&
        H_{3}-0-\int_{[0,1]} \hat{\tilde{G}}_{4}+\tfrac{1}{2} \int_{[0,1]} \hat{p}_{1} 
        \\
		&=&
        H_{3}-C_{3}+\tfrac{1}{2} \int_{[0,1]} \hat{p}_{1}
        \\
        &=&
        H_{3}-C_{3}+CS(\omega)\,.
	\end{array}
\end{equation}
Thus, the expressions in equation \eqref{FluxDensitiesFromNullConcordances} are all consistent.

Next, to show that this construction is surjective, we claim that for given gauge potentials $(B_2, C_3, C_6)$ the following is a null concordance that recovers these gauge potentials via \eqref{FluxDensitiesFromNullConcordances}:
\vspace{-2mm} 
\begin{equation}
  \label{NullConcordancesForGivenGaugPotentials}
  \renewcommand{\arraystretch}{1.5}
	\begin{array}{rcl}
		\hat{p}_{1}
          &=&
        t{p_{1}}+\dd t \, s_{3} 
        \\
        \hat{G}_{4}
        &=& 
        t G_{4}+\dd t\, (C_{3}-\tfrac{1}{4}s_{3}), \quad \hat{p}_{1}=t {p_{1}}+\dd t\, s_{3}
        \\
        \hat{\tilde{G}}_{4}
        &=&
			t (\left(G_{4}+\tfrac{1}{4}p_{1}\right)+\dd t\, C_{3} = t \tilde{G_{4}}+ \dd t\, C_{3}
        \\ 
        \hat{G}_{7}
          &=&
         t^{2} G_{7}+2 t\, \dd t\, C_{6}+\tfrac{1}{4}t \dd t C_{3} s_{3}
         \\
        \hat{H}_{3} &=& t H_{3}+\dd t\, B_{2}\,.
		\end{array}
        \end{equation}

We need to check 

\begin{itemize}[
  leftmargin=2cm
]
\item[(1.)] that \eqref{NullConcordancesForGivenGaugPotentials} recovers the given gauge potentials when plugged into the formulas \eqref{FluxDensitiesFromNullConcordances},

\item[(2.)] that \eqref{NullConcordancesForGivenGaugPotentials} are actually null concordances for the given flux densities.
\end{itemize}

Keeping in mind that inside the $t$-integrals, only the terms having $\dd t$ give non-zero contributions, one can check that the expressions listed above reproduce the usual forms of the gauge potentials mentioned in \eqref{FluxDensitiesFromNullConcordances}, and also the usual null concordance conditions at $t=0$ and $t=1$. So the expressions are consistent.

\smallskip

Now, we go ahead to check that these expressions satisfy the usual Bianchi identities in \eqref{TheFluxDensities}.

We start with $\hat{{G}}_{4}$,
        \begin{equation}
            \begin{array}{l}

		\dd \,\hat{G}_{4}=\dd t\, G_{4}-\dd\, t \tilde{G_{4}}+\tfrac{1}{4}\dd t \,p_{1}=0 \\
	
	\end{array}
\end{equation}
Next, we want to check $\hat{p}_{1}$
        \begin{equation}
            \dd\hat{p}_{1}=\dd t\, p_{1}-\dd t\, \dd s_{3}=0 \,( ds_{3}=p_{1}).
        \end{equation}
        Next, we check $\hat{\tilde{G}}_{4}$
        \begin{equation}
		\dd \hat{\tilde{G_{4}}}=\dd t\, \tilde{G}_{4}-\dd t\, \tilde{G}_{4}=0\,.
\end{equation}
Now, moving onto $\hat{H}_{3}$, we see
\begin{align}
\dd \hat{H}_{3}&=\dd t\, H_{3}+t\left(G_{4}-\tfrac{1}{4} p_{1}\right)-\dd t\left(H_{3}-C_{3}+
\tfrac{1}{2}s_{3}\right)=t\left(G_{4}-\tfrac{1}{4} p_{1}\right)+\dd t\left(C_{3}-
\tfrac{1}{2}s_{3}\right),
\\
\hat{G}_{4}-\tfrac{1}{4} \hat{p}_{1}&=t G_{4}+\dd t\left(c_{3}-\tfrac{1}{4}s_{3}\right)-\tfrac{t}{4} p_{1}-\tfrac{\dd t}{4} s_{3}=t\left(G_{4}-\tfrac{1}{4} p_{1}\right)+\dd t\left(C_{3}-\tfrac{1}{2}s_{3}\right).
        \end{align}
        
		\text { Thus, we recover the expected Bianchi identity } 
        \begin{equation}\dd \hat{H}_{3}=\hat{G}_{4}-\tfrac{1}{4} \hat{p}_{1}\end{equation}

	Now, $$\tilde{G}_{7}:=G_{7}-\tfrac{1}{2} C_{3}\left(\tilde{G}_{4}-\tfrac{1}{2} p_{1}\right),\qquad \dd C_{6}=\tilde{G}_{7},\qquad \dd \tilde{G}_{7}=0 $$
    Also,
		\begin{equation}
		    \hat{C}_{3}=\int_{[0,t]} \hat{\tilde{G}}_{4}=\int_{[0,t]} t^{\prime}\tilde{G}_{4}+\dd t'\, C_{3}=t C_{3}\,,
            \end{equation}
            $$\iota^{*}_{0}\hat{C}_{3}=0,\iota^{*} _{1} \hat{C}_{3}=C_{3}\,,
            $$
$$d\hat{C_{3}}= \dd \,t C_{3}+ t \tilde{G_{4}}= \hat{\tilde{G}}_{4}
\,,$$
\begin{equation}
	\begin{array}{c}
		\hat{\tilde{G}}_{7}=t^{2} \tilde{G}_{7}+2 t \dd t\, C_{6} \,,\\
		C_{6}=\int_{[0,1]} \hat{\tilde{G}}_{7}\,,
        \\
        \dd C_{6}=\iota_{1}^{*} \hat{\tilde{G}}_{7}-\iota_{0}^{*} \hat{\tilde{G}}_{7}-\int_{[0,1]} \dd \hat{\tilde{G}}_{7}=\left[G_{7}-\tfrac{1}{2} C_{3}\left(\tilde{G}_{4}-\tfrac{1}{2} p_{1} \right)\right]-0-0= \tilde{G}_{7}\,, \\
	\end{array}
\end{equation}

\begin{equation}
	\begin{aligned}
		\dd \hat{\tilde{G}}_{7} & =2 t \dd t\, \tilde{G}_{7}+t^{2} \dd \tilde{G}_{7}-2 t \dd t\, \dd C_{6} \\
		& =2 t \dd t \widetilde{G}_{7}-2 t \dd t \tilde{G}_{7}+0=0\,.
	\end{aligned}
\end{equation}
Now, using the above-listed expressions, we want to get to $\hat{G}_{7}$
\begin{equation}
\renewcommand{\arraystretch}{1.5}
	\begin{array}{l}
		\hat{\tilde{G}}_{7}=t^{2} \tilde{G}_{7}+2 t \dd t\, C_{6} \\
		\hat{G}_{7}-\tfrac{1}{2} \hat{C}_{3}\left(\hat{G}_{4}-\tfrac{1}{4} \hat{p}_{1}\right)=t^{2}\left(G_{7}-\tfrac{1}{2} C_{3}\big(\tilde{G}_{4}-\tfrac{1}{2} p_{1}\big)\right)+2 t \dd t\, C_{6} \\
		\hat{G}_{7}-\tfrac{1}{2} t C_{3}\left[t\left(G_{4}-\tfrac{1}{4} p_{1}\right)+\dd t\left(C_{3}-\tfrac{1}{2} s_{3}\right)\right]=t^{2} G_{7}-\tfrac{t^{2}}{2} C_{3}\left(G_{4}-\tfrac{1}{4} p_{1}\right)+2 t \dd t\, C_{6} \\
		\hat{G}_{7}-\tfrac{t^{2}}{2} C_{3}\left(G_{4}-\tfrac{1}{4} p_{1}\right)-\tfrac{1}{4}t \dd t C_{3} s_{3}=t^{2} G_{7}-\tfrac{t^{2}}{2} C_{3}\left(G_{4}-\tfrac{1}{4} p_{1}\right)+2 t \dd t\, C_{6} \\
		\hat{G}_{7}=t^{2} G_{7}+2 t \dd t\, C_{6}+\tfrac{1}{4} t \dd t\, C_{3} s_{3}
	\end{array}
\end{equation}
as advertised above.

Now, we check 
\begin{equation}
\renewcommand{\arraystretch}{1.4}
\begin{array}{ll}
    \dd\hat{G}_{7}&= \tfrac{1}{2}t^{2} (G_{4}+\tfrac{1}{4}p_{1})(G_{4}-\tfrac{1}{4}p_{1})+ t \dd \,t C_{3}(G_{4}-\tfrac{1}{4}p_{1}) + \tfrac{1}{4}t\dd t\, C_{3} p_{1}- \tfrac{1}{4}t \dd t\,(G_{4}+\tfrac{1}{4}p_{1}) s_{3} 
    \\
   & = \tfrac{1}{2}(\hat{G}_{4}+\tfrac{1}{4}\hat{p}_{1})(\hat{G}_{4}-\tfrac{1}{4}\hat{p}_{1})
    \end{array}
\end{equation} 
as expected.
\end{proof}

This completes the discussion of null concordances.

\medskip 
Let us turn to the concordances of concordances now.
  We claim below the forms of the concordances of concordances between the null concordances have surjective maps to the usual forms of differential forms, which encode gauge transformations of the gauge potentials.
\begin{proposition} [\bf Null-concordances of concordances as producing gauge equivalences]
\begin{equation}
\label{Null-concordances of concordances as producing gauge equivalences}
\renewcommand{\arraystretch}{1.5}
\begin{array}{l}
	C_{2}=\int_{s \in{[0,1]}} \int_{t \in[0,1)} \hat{\hat{\tilde{G}}}_{4}\\ B_{1}=\int_{s \in{[0,1]}} \int_{t \in{[0,1]}} {\hat{\hat{H}}}_{3}\\C_{5}=\int_{s \in{[0,1]}}\int_{ t \in{[0,1]}} (\hat{\hat{{G}}}_{7}-\tfrac{1}{2}\hat{\hat{C}}_{3}(\hat{\hat{\tilde{G}}}_{4}-\tfrac{1}{2}\hat{\hat{p}}_{1}))-\tfrac{1}{2} C_{2} C_{3}\\
    s_{2}=\int_{s \in{[0,1]}}\int_{ t \in{[0,1]}} \hat{\hat{{p}}}_{1}\,.
    \end{array}
\end{equation}
	
    \end{proposition}

\begin{proof}
    Let's check the consistency of these relations to satisfy the expected Bianchi identities as mentioned in \eqref{GaugeTransformations}.

  Firstly, we check the $C_{2}$ relation.
  \begin{equation}
  \renewcommand{\arraystretch}{2.1}
  \begin{array}{ll}
   \dd C_{2}&=\dd \int_{s\in{[0,1]}} \int_{t\in{[0.1]}}{\hat{\hat{\tilde{G}}}}_{4}=\iota^{*}_{s=1}\int_{t \in{[0,1]}} {\hat{\hat{\tilde{G}}}_{4}}-\iota^{*}_{s=0} \int_{t \in{[0,1]}}{\hat{\hat{\tilde{G}}}_{4}}-\int_{s \in{[0,1]}} \dd \int_{t \in{[0,1]}} {\hat{\hat{\tilde{G}}}}_{4}
\\[2pt]
	&=\int_{t\in[0,1]} \iota^{*}_{s=1} \hat{\hat{\tilde{G}}}_{4}-\int_{t \in[0,1]} \iota^{*}_{s=0} {\hat{\hat{\tilde{G}}}}_{4}-\int_{s \in[0,1]}\left(\iota^{*}_{t=1} \hat{\hat{\tilde{G}}}_{4}-\iota^{*}_{t=0} \hat{\hat{\tilde{G}}}_{4}\right)+\int_{s \in[0,1]}\int_{ t\in{[0,1]}} \dd {\hat{\hat{\tilde{G}}}}_{4}
\\
&=\int_{t\in[0,1]}\hat{\tilde{G}}^{\prime}_{4}- \int_{t\in[0,1]}\hat{\tilde{G}}_{4}-0-0
\\
&	=C^{\prime}_{3}-C_{3}.
\end{array}
\end{equation}
Next, moving onto the $B_{1}$ equation, we see
\begin{equation}
  \renewcommand{\arraystretch} {2} 
\begin{array}{ll}
  \dd B_{1}=\dd \int_{s\in{[0,1]}} \int_{t\in{[0,1]}}{\hat{\hat{H}}}_{3}=\iota^{*}_{s=1}\int_{t \in{[0,1]}} {\hat{\hat{H}}_{3}}-\iota^{*}_{s=0} \int_{t \in{[0,1]}}{\hat{\hat{H}}_{3}}-\int_{s \in{[0,1]}} \dd \int_{t \in{[0,1]}} {\hat{\hat{H}}}_{3}\\

	=\int_{t\in[0,1]} \iota^{*}_{s=1} \hat{\hat{H}}_{3}-\int_{t \in[0,1]} \iota^{*}_{s=0} {\hat{\hat{H}}}_{3}-\int_{s \in[0,1]}\left(\iota^{*}_{t=1} \hat{\hat{H}}_{3}-\iota^{*}_{t=0} \hat{\hat{H}}_{3}\right)+\int_{s \in[0,1]}\int_{ t\in{[0,1]}} \dd {\hat{\hat{H}}}_{3}\\
    =\int_{t\in[0,1]}\hat{H^{\prime}_{3}}- \int_{t\in[0,1]}\hat{H}_{3}-0+\int_{s\in[0,1]}\int_{t\in[0,1]}\hat{\hat{\tilde{G}}}_{4}-\tfrac{1}{2}\int_{s\in[0,1]}\int_{t\in[0,1]}\hat{\hat{p}}_{1}\\
    
	=B_{2}{ }^{\prime}-B_{2}+C_{2}-\tfrac{1}{2} \int_{s \in[0, 1]} \int_{t \in[0,1]}{\hat{\hat{p}}_{1}}\\
    =B_{2}{ }^{\prime}-B_{2}+C_{2}-\tfrac{1}{2}s_{2}\,.
    \end{array}
\end{equation}

Next, we want to check the $C_{5}$ relation.
\begin{align}
	\hat{\hat{\tilde{G}}}_{7} & :={\hat{\hat{G}}}_{7}-\tfrac{1}{2} \hat{\hat{{C}}}_{3}\Big(\hat{\hat{\tilde{G}}}_{4}-\tfrac{1}{2} {\hat{\hat{p}}}_{1}\Big)
\\
	C_{5} &:= \int_{s \in[0, 1]} \int_{t \in[0,1]}
    \left({\hat{\hat{G}}}_{7}-\tfrac{1}{2}{\hat{\hat{C}}}_{3}\Big({\hat{\hat{\tilde{G}}}}_{4}-\tfrac{1}{2} {\hat{\hat{p}}}_{1}\Big)\right)-\tfrac{1}{2} C_{2} C_{3}\\
	&=\int_{s \in[0,1)} \int_{t \in[0,1]} \hat{\hat{\tilde{G}}}_{7}-\tfrac{1}{2} C_{2} C_{3}\;.
\end{align}

Now,	\begin{equation}
    \renewcommand{\arraystretch}{2}
    \begin{array}{l}
\dd C_{5}=\dd \int_{s \in[0,1]} \int_{t \in[0,1]} \hat{\hat{\tilde{G}}}_{7}-\tfrac{1}{2} \dd C_{2} C_{3}-\tfrac{1}{2} C_{2} \dd C_{3}\\
=\iota^{*}_{s=1}\int_{t \in{[0,1]}} {\hat{\hat{\tilde{G}}}_{7}}-\iota^{*}_{s=0} \int_{t \in{[0,1]}}{\hat{\hat{\tilde{G}}}_{7}}-\int_{s \in{[0,1]}} \dd \left(\int_{t \in{[0,1]}} {\hat{\hat{G}}}_{7}-\tfrac{1}{2}{\hat{\hat{C}}}_{3}\Big({\hat{\hat{\tilde{G}}}}_{4}-\tfrac{1}{2} {\hat{\hat{p}}}_{1}\Big)\right)
-\tfrac{1}{2}(C^{\prime}_{3}-C_{3})C_{3}-\frac{1}{2}C_{2} \tilde{G_{4}}\\

	=\int_{t\in[0,1]} \hat{\tilde{G}}^{\prime}_{7}-\int_{t \in[0,1]} {\hat{\tilde{G}}}_{7}+\tfrac{1}{2}\Big(\int_{s \in[0,1]}\int_{t\in[0,1]}\hat{\hat{\tilde{G}}}_{4}\Big)\big(\tilde{G}_{4}-\tfrac{1}{2}p_{1}\big)-\tfrac{1}{2} C_{3}{ }^{\prime} C_{3}-\tfrac{1}{2} C_{2} \tilde{G}_{4}
\\
		=C_{6}{ }^{\prime}-C_{6}+\tfrac{1}{2} C_{2}\left(\tilde{G}_{4}-\tfrac{1}{2} p_{1}\right)-\tfrac{1}{2} C_{3}{ }^{\prime} C_{3}-\tfrac{1}{2} C_{2} \tilde{G}_{4} \\
		=C_{6}{ }^{\prime}-C_{6}-\tfrac{1}{2} C_{3}{ }^{\prime} C_{3}-\tfrac{1}{4} C_{2} p_{1}\;.
	\end{array}
\end{equation}
Finally, let's check the $s_{2}$ equation.
\begin{equation}
\renewcommand{\arraystretch}{1.8}
    \begin{array}{l}
       \dd{s}_{2}=\dd \int_{s\in{[0,1]}} \int_{t\in{[0.1]}}{\hat{\hat{{p}}}}_{1}=\iota^{*}_{s=1}\int_{t \in{[0,1]}} {\hat{\hat{{p}}}_{1}}-\iota^{*}_{s=0} \int_{t \in{[0,1]}}{\hat{\hat{{p}}}_{1}}-\int_{s \in{[0,1]}} \dd \int_{t \in{[0,1]}} {\hat{\hat{{p}}}}_{1}\\

	=\int_{t\in[0,1]} \iota^{*}_{s=1} \hat{\hat{{p}}}_{1}-\int_{t \in[0,1]} \iota^{*}_{s=0} {\hat{\hat{{p}}}}_{1}-\int_{s \in[0,1]}\left(\iota^{*}_{t=1} \hat{\hat{{p}}}_{1}-\iota^{*}_{t=0} \hat{\hat{{p}}}_{1}\right)+\int_{s \in[0,1]}\int_{ t\in{[0,1]}} \dd {\hat{\hat{p}}}_{1}\\
    =\int_{t\in[0,1]}\hat{{p}}^{\prime}_{1}- \int_{t\in[0,1]}\hat{{p}}_{1}-0-0\\=s_{3}^{\prime}- s_{3}\,.
 
    \end{array}
\end{equation}

\vspace{-7mm} 
\end{proof}

Let us show the surjections explicitly and check their consistency.

\begin{proposition} [\bf Surjections]
    The surjections are given by setting:
\begin{equation}
\renewcommand{\arraystretch}{1.6}
    \begin{array}{l}
        \hat{\hat{p}}_{1}=t p_{1}+\dd t\, s_{3}+s \dd t\left(s_{3}^{\prime}-s_{3}\right)-\dd s\, \dd t\, s_{2}\\
        {\hat{\hat{\tilde{G}}}}_{4}=t{\tilde{G}_{4}}+\dd t \,C_{3}+s \dd t\left(C_{3}^{\prime}-C_{3}\right)-\dd s\, \dd t\, C_{2}\\
        {\hat{\hat{G}}}_{7}
		= t^{2} G_{7}+2 t\dd t\, C_{6}+ \tfrac{1}{4} t\dd t\, C_{3}s_{3}+ 2 s t \dd t\left(C_{6}^{\prime}-C_{6}\right)-2 t\dd s\, \dd t\left(C_{5}+\tfrac{1}{2} C_{2} C_{3}\right) \\
		\qquad -  \tfrac{1}{4}s\, t\, \dd t C_{3}\left(s_{3}-s_{3}^{\prime}\right)+\tfrac{1}{4}s\, t\, \dd t\left(C_{3}^{\prime}-C_{3}\right) s_{3}-\tfrac{1}{4} s^{2}\, t\,\dd t\left(C_{3}^{\prime}-C_{3}\right)\left(s_{3}-s_{3}^{\prime}\right) \\
		 \qquad +\tfrac{1}{4}t\, \dd s\, \dd t \left(C_{3} s_{2}-C_{2} s_{3}\right)+\tfrac{1}{4} st\, \dd s\, \dd t\left(C_{3}^{\prime}-C_{3}\right) s_{2}-\tfrac{1}{4}s t\, \dd s\, \dd tC_{2}\left(s_{3}^{\prime}-s_{3}\right)\\
         \hat{\hat{H}}_{3}= t H_{3}+\dd t\, B_{2}+s \dd t\left(B_{2}^{\prime}-B_{2}\right)-\dd s\, \dd t\, B_{1}\,.
    \end{array}
\end{equation}
\end{proposition}

\begin{proof}
We can immediately see that the expressions mentioned above are consistent when plugged in \eqref{Null-concordances of concordances as producing gauge equivalences} and the usual null-concordance of concordance relations at $s=0,1; t=0,1$.

Now, we want to check that they satisfy the usual Bianchi identities.

 We start with
\begin{equation}
	\hat{\hat{p}}_{1}=t p_{1}+\dd t s_{3}+s \dd t\left(s_{3}^{\prime}-s_{3}\right)-\dd s \dd t s_{2}
\end{equation}
and using \begin{equation}\dd s_{2}=s_{3}^{\prime}-s_{3}
\end{equation} we find
\begin{equation}
	\dd \hat{\hat{p}}_{1}
	=\dd t\, {p_{1}}-\dd t\, p_{1}+\dd s\, \dd t\left(s_{3}{ }^{\prime}-s_{3}\right)-s d t .0-\dd s\, \dd t\ \dd s_{2}=0.
    \end{equation}

Next, we move onto $\hat{\hat{G}}_{4}$ and $\hat{\hat{\tilde{G}}}_{4}$. 

\begin{equation}
	{\hat{\hat{G}}}_{4}=t{G_{4}}+\dd t\left(C_{3}-\tfrac{1}{4}s_{3}\right)+s \dd t\left(C_{3}^{\prime}-\tfrac{1}{4}s_{3}^{\prime}-C_{3}+\tfrac{1}{4}s_{3}\right)-\dd s\, \dd t\,\left(C_{2}-\tfrac{1}{4} s_{2}\right)
\end{equation}

\begin{equation}
	\begin{aligned}
	\hat{\hat{\tilde{G}}}_{4} &=   t\left(G_{4}+\tfrac{1}{4} p_{1}\right)+\dd t\, C_{3}+s \dd t\left(C_{3}^{\prime}-C_{3}\right)-\dd s\, \dd t\, C_{2} \\
		& =t \tilde{G}_{4}+\dd t\, C_{3}+s \dd t\left(C_{3}^{\prime}-C_{3}\right)-\dd s\, \dd t\, C_{2}\,.
	\end{aligned}
\end{equation}
Now, we check
\begin{equation}
	\mathrm{d}\hat{\hat{\tilde{G}}}_{4}=\dd t\, \tilde{G_{4}}-\dd t\, \tilde{G_{4}}-s \dd t\big(\tilde{G}_{4}-\tilde{G}_{4}\big)+\dd s\,\d\,d t\left(C_{3}{ }^{\prime}-C_{3}\right)-\dd s\, \dd t\,\left(C_{3} ^{\prime}-C_{3}\right)=0\,,
\end{equation}
and also the same for $\hat{\hat{{G}}}_{4}$
\begin{align}
\hspace{-8mm} 
	\dd \hat{\hat{G}}_{4}&=\dd t\, G_{4}-\dd t\, G_{4}+s \dd t\left(C_{3}^{\prime}-\tfrac{1}{4}s_{3}^{\prime}-C_{3}+\tfrac{1}{4}s_{3}\right)-s \dd t\left(G_{4}-G_{4}\right)
    -\dd s\, \dd t\,\left(C_{3}^{\prime}-C_{3}-\tfrac{1}{4}{s_{3}^{\prime}}+\tfrac{1}{4}{s_{3}}\right)
    \nonumber
    \\
    &=0.
\end{align}

Next, moving onto $\hat{\hat{H}}_{3}$ we see
\begin{equation}
	\hat{\hat{H}}_{3}=t H_{3}+\dd t\, B_{2}+s \dd t\left(B_{2}^{\prime}-B_{2}\right)-\dd s\, \d\,d t B_{1}
\end{equation}
and we check the corresponding Bianchi identity
\begin{equation*}
\begin{array}{l}
\dd \hat{\hat{H}}_{3}=\dd t\, H_{3}-\dd t\,\left(H_{3}-C_{3}+\tfrac{1}{2}s_{3}\right)+t\left(G_{4}-\tfrac{1}{4} p_{1}\right)+\dd s\, \dd t\,\left(B_{2}^{\prime}-B_{2}\right)\\[4pt]
   \qquad \quad   -s \dd t\left(\tfrac{1}{2}{s_{3}^{\prime}}-C_{3}^{\prime}-\tfrac{1}{2}s_{3}+C_{3}\right)-\dd s\, \dd t\, \dd B_{1}
    \end{array}
\end{equation*}
so
    \renewcommand{\arraystretch}{2}
	\begin{align}
    \dd \hat{\hat{H}}_{3}&=t\left(G_{4}-\tfrac{1}{4} p_{1}\right)+\dd t\left(C_{3}-\tfrac{1}{2}{s_{3}}\right)+s \dd t\left(C_{3}^{\prime}-C_{3}-\tfrac{1}{2}s_{3}^{\prime}+\tfrac{1}{2}s_{3}\right)-\dd s\, \dd t\,\left(C_{2}-\tfrac{1}{2}s_{2}\right)
\nonumber 
\\
	&=\hat{\hat{{G}}}_{4}-\tfrac{1}{4} \hat{\hat{p}}_{1} =\hat{\hat{\tilde{G}}}_{4}-\tfrac{1}{2} {\hat{\hat{p}}}_{1}
\end{align}

as expected.

Finally, we want to check the case for $\hat{\hat{G}}_{7}$.
\begin{equation}
	\begin{array}{r}
		\hat{\hat{\tilde{{G}}}}{ }_{7}:=t^{2} \tilde{G}_{7}+2 t \dd t C_{6}+2 s t\, \dd t\left(C_{6}^{\prime}-C_{6}\right)-2 \dd s\, t\,\dd t\left(C_{5}+\tfrac{1}{2} C_{2} C_{3}\right)+s t \dd t C_{2}\left(\tilde{G}_{4}-\tfrac{1}{2}p_{1}\right) \\
		+\dd\left(-\tfrac{s t^{2}}{2} C_{2}\left(\tilde{G}_{4}-\frac{1}{2}p_{1}\right)\right).
	\end{array}
\end{equation}

We check

       \begin{equation}
\renewcommand{\arraystretch}{1.5}
	\begin{array}{ll}
		\dd \hat{\hat{\tilde{G}}}_{7}& \!\!\!=2 t \dd t \tilde{G}_{7}+t^{2} \dd \tilde{G}_{7}-2 t \dd t \tilde{G}_{7}+2 \dd s\, t\, \dd t\left(C_{6}^{\prime}-C_{6}\right)-2 s t \dd t\left\{\left(G_{7}-\tfrac{1}{2} C_{3}^{\prime}\left(\tilde{G}_{4}-\tfrac{1}{2} p_{1}\right)\!\right)\!\right. \\
		&\!\!\!\! \left.-\left(G_{7}-\tfrac{1}{2} C_{3}\left(\tilde{G}_{4}-\tfrac{1}{2} p_{1}\right)\!\right)\!\right\}-2 \dd s\, t\, \dd t\left[\left(C_{6}^{\prime}-C_{6}-\tfrac{1}{2} C_{3}^{\prime} C_{3}-  \tfrac{1}{4} C_{2}p_{1}\right)+\tfrac{1}{2}\!\left\{\left(C_{3}^{\prime}-C_{3}\right) \cdot C_{3}+C_{2} \widetilde{G}_{4}\right\}\right] \\
		\\[-15pt]
		&
        \!\!\! +\dd s\, t\, \dd t C_{2}\left(\tilde{G}_{4}-\tfrac{1}{2} p_{1}\right)-s t\, \dd t\left(C_{3}^{\prime}-C_{3}\right)\left(\tilde{G}_{4}-\tfrac{1}{2} p_{1}\right)+0 
		\\[7pt]
		&=0\,.
	\end{array}
\end{equation}

We note that,
	$$s t\, \dd t C_{2}\left(\tilde{G}_{4}-\tfrac{1}{2} p_{1}\right)+\dd\left(-\tfrac{s t^{2}}{2} C_{2}\left(\tilde{G}_{4}-\tfrac{1}{2} p_{1}\right)\right)=\\-\frac{s t^{2}}{2}\left(C_{3}^{\prime}-C_{3}\right)\left(\tilde{G}_{4}-\tfrac{1}{2} p_{1}\right)-\dd s \tfrac{t^{2}}{2} C_{2}\left(\tilde{G}_{4}-\tfrac{1}{2} p_{1}\right).
$$
Thus, 
\begin{equation}
\renewcommand{\arraystretch}{2}
	\begin{array}{l}
		\hat{\hat{\tilde{G}}}_{7}=t^{2}\left(G_{7}-\tfrac{1}{2} C_{3}\left(\tilde{G}_{4}-\tfrac{1}{2} p_{1}\right)\right)+2 t \dd t C_{6}+2 s t\, \dd t\left(C_{6}^{\prime}-C_{6}\right)-2 \dd s\, t\,\dd t\left(C_{5}+\tfrac{1}{2} C_{2} C_{3}\right) 
        \\
	 \qquad	-\tfrac{1}{2}s t^{2}\left(C_{3}^{\prime}-C_{3}\right)\left(\tilde{G}_{4}-\tfrac{1}{2} p_{1}\right)-\tfrac{1}{2}\dd s\, t^{2} C_{2}\left(\tilde{G}_{4}-\frac{1}{2} p_{1}\right).
	\end{array}
\end{equation}
Now, we want to obtain $\hat{\hat{G}}_{7}$ from  $\hat{\hat{\tilde{G}}}_{7}$
\begin{equation*}
	{\hat{\hat{\tilde{G}}}}_{7}=\hat{\hat{{G}}}_{7}-\tfrac{1}{2} {\hat{\hat{C}}}_{3}\Big({\hat{\hat{\tilde{G}}}}_{4}-\tfrac{1}{2} {\hat{\hat{p}}}_{1}\Big).
\end{equation*}
Also,
\begin{eqnarray}\
		\hat{\hat{C}}_{3}&=&\int_{t^{\prime} \in[0, t]} \hat{\hat{\tilde{G}}}_{4} =\int_{t^{\prime} \in[0, t]} t^{\prime} \tilde{G}_{4}+\dd t^{\prime} C_{3}+s\, \dd t^{\prime}\left(C_{3}^{\prime}-C_{3}\right)-\dd s\, \dd t^{\prime}\, C_{2} 
        \nonumber 
        \\[3pt]
		& =& t C_{3}+s t\left(C_{3}^{\prime}-C_{3}\right)+t \dd s\, C_{2},
\\
       d\hat{\hat{C}}_{3}&=& t \tilde{G}_{4}+ \dd t\, C_{3}+ s \dd \,t (C^{\prime}_{3}-C_{3})- \dd s\, \dd t\, C_{2}= \hat{\hat{\tilde {G}}}_{4}\,.
   \end{eqnarray}
   
   And,
\begin{equation*}
	\begin{array}{ll}
		\hat{\hat{\tilde{G}}}_{4}-\tfrac{1}{2} \hat{\hat{p}}_{1}&={\hat{\hat{G}}}_{4}-\tfrac{1}{4} \hat{\hat{p}}_{1} \\
		&=t\left(G_{4}-\tfrac{1}{4} p_{1}\right)+\dd t\left(C_{3}-\tfrac{1}{2}s_{3}\right)+s \dd t\left(C_{3}^{\prime}-\tfrac{1}{2}s_{3}^{\prime}-C_{3}+\tfrac{1}{2}s_{3}\right)-\dd s\, \dd t\,\left(C_{2}-\tfrac{1}{2}s_{2}\right).
	\end{array}
\end{equation*}
Thus,
  \begin{equation}
\renewcommand{\arraystretch}{1.6}
  \begin{array}{ll}
{\hat{\hat{C}}_{3}}\left(\hat{\hat{G}}_{4}-\tfrac{1}{4} {\hat{\hat{p}}_{1}}\right) \\
=(t C_{3}+s t(C_{3}^{\prime}-C_{3})+t \dd s \, C_{2}) [t(G_{4}-\tfrac{1}{4} p_{1})+\dd t\,(C_{3}-\tfrac{1}{2}s_{3})+s \dd t\,(C_{3}^{\prime}-\tfrac{s_{3}^{\prime}}{2}-C_{3}+\tfrac{1}{2}s_{3}) \\ -\dd s\, \dd t\,(C_{2}-\tfrac{1}{2}s_{2})]
\\[5pt]	
=t^{2} C_{3}\left(G_{4}-\tfrac{1}{4} p_{1}\right)-t \dd t\, C_{3}\left(C_{3}-\tfrac{1}{2}s_{3}\right)-t s\, \dd t C_{3}\left(C_{3}^{\prime}-C_{3}-\tfrac{1}{2}s_{3}^{\prime}+\tfrac{1}{2}s_{3}\right)-t \dd s\, \dd t\, C_{3}\left(C_{2}-\tfrac{1}{2}s_{2}\right) \\
+s t^{2}\left(C_{3}^{\prime}-C_{3}\right)\left(G_{4}-\tfrac{1}{4} p_{1}\right)-s\, t\, \dd t\left(C_{3}^{\prime}-C_{3}\right)\left(C_{3}-\tfrac{1}{2}s_{3}\right)-s^{2} t \,\dd t\left(C_{3}^{\prime}-C_{3}\right)\left(\tfrac{s_{3}}{2}-\tfrac{1}{2}s_{3}^{\prime}\right) \\ -s t\, \dd s\,\dd t\left(C_{3}^{\prime}-C_{3}\right)\left(C_{2}-\tfrac{1}{2}s_{2}\right)+t^{2} \,\dd s C_{2}\left(G_{4}-\tfrac{1}{4} p_{1}\right)+t\, \dd s\, \dd t\, C_{2}\left(C_{3}-\tfrac{1}{2}s_{3}\right)\\+s t\, \dd s\, \dd t\, C_{2}\left(C_{3}^{\prime}-C_{3}-\tfrac{1}{2}s_{3}^{\prime}+\tfrac{1}{2}s_{3}\right)
\\[6pt]
=t^{2} C_{3}\left(G_{4}-\tfrac{1}{4} p_{1}\right)- t\, \dd t\, C_{3} (C_{3}-\tfrac{1}{2}s_{3})+s t^{2}\left(C_{3}^{\prime}-C_{3}\right)\left(G_{4}-\tfrac{1}{4} p_{1}\right)-s t\, \dd t\, C_{3}\left(\tfrac{1}{2}s_{3}-\tfrac{1}{2}s_{3}^{\prime}\right)\\+\tfrac{1}{2}s t\, \dd t\left(C_{3}^{\prime}-C_{3}\right) s_{3}
-\tfrac{1}{2}s^{2} t\, \dd t\left(C_{3}^{\prime}-C_{3}\right) \left(s_{3}-s_{3}^{\prime}\right)+t^{2} \dd s \,C_{2}\left(G_{4}-\tfrac{1}{4} p_{1}\right)+\tfrac{1}{2}t \, \dd s\,\dd t\left(C_{3} s_{2}-C_{2} s_{3}\right) \\
-\tfrac{1}{4}s t\, \dd s\, \dd t \quad\left(C_{3}^{\prime}-C_{3}\right)\left(-{s_{2}}\right)s\, t\, \dd s\, \dd t\, C_{2}\left(s_{3}^{\prime}-s_{3}\right).
\end{array}
\end{equation}
Now, we get 
\begin{equation}
\renewcommand{\arraystretch}{2}
	\begin{aligned}
		{\hat{\hat{G}}}_{7} =&  \; {\hat{\hat{\tilde{G}}}}_{7}+\tfrac{1}{2}{\hat{\hat{C}}}_{3}\left(\hat{\hat{G}}_{4}-\tfrac{1}{4} \hat{\hat{p}}_{1}\right) \\
		=& \; t^{2} G_{7}+2 t\, \dd t C_{6}+ \tfrac{1}{4} t\, \dd \,t C_{3}s_{3}+ 2 s t\, \dd t\left(C_{6}^{\prime}-C_{6}\right)-2t \, \dd s\, \dd t\,\left(C_{5}+\tfrac{1}{2} C_{2} C_{3}\right) \\
		&-  \tfrac{1}{4}s t \, \dd t C_{3}\left(s_{3}-s_{3}^{\prime}\right)+\tfrac{1}{4}s t\, \dd t\left(C_{3}^{\prime}-C_{3}\right) s_{3}-\tfrac{1}{4} s^{2} t\, \dd t\,\left(C_{3}^{\prime}-C_{3}\right)\left(s_{3}-s_{3}^{\prime}\right) \\
		& +\tfrac{1}{4}t\, \dd s\,\dd t\left(C_{3} s_{2}-C_{2} s_{3}\right)+\tfrac{ 1}{4}st\,\dd s\, \dd t\left(C_{3}^{\prime}-C_{3}\right) s_{2}-\tfrac{1}{4}s\, t\, \dd s\, \dd t C_{2}\left(s_{3}^{\prime}-s_{3}\right).
	\end{aligned}
\end{equation}
Finally, we want to check its Bianchi identity.
 \begin{equation}
\renewcommand{\arraystretch}{1.6}
     \begin{array}{ll}
         \dd \hat{\hat{{G}}}_{7}&= \frac{1}{2}t^{2} (G_{4}+\tfrac{1}{4}p_{1})(G_{4}-\tfrac{1}{4}p_{1})+ t\, \dd t\left( C_{3} G_{4} - \tfrac{1}{4}s_{3} (G_{4}+ \tfrac{1}{4}p_{1})\right) 
         \\ 
         &\;\;\;+st\, \dd t \left((C^{\prime}_{3}-C_{3})G_{4} + \tfrac{1}{4} (G_{4}+\tfrac{1}{4}p_{1}) (s_{3}- s^{\prime}_{3})\right) + t\, \dd s\, \dd t\, (-C_{2}G_{4}+ \tfrac{1}{4}s_{2} (G_{4}+\tfrac{1}{4}p_{1}))\\
         &= \tfrac{1}{2}\hat{\hat{\tilde{G}}}_{4}(\hat{\hat{\tilde{G}}}_{4}-\tfrac{1}{2}\hat{\hat{p}}_{1}).
         \end{array}
 \end{equation} 

 \vspace{-5mm} 
\end{proof}


\noindent  \textbf{Note}: If we set $p_{1}, s_{3}, s_{2} = 0$, we will exactly reproduce the results known earlier in \cite[\S 2.1.4]{GSS24-SuGra}\cite[\S 4.1]{GSS24-FluxOnM5}.
\begin{remark}
We should keep in mind that $\hat{\hat{\tilde{G}}}_{7}$ is {\it not} a concordance of concordance between the pair $\hat{\tilde{G}}_{7}$ and $\hat{\tilde{G}}^{\prime}_{7}$ for the same $\tilde{G}_{7}$, as it contains the potential term $C_{3}$ which is {\it not} topological. We have done the calculations by just defining something to be like that, to derive the explicit expressions for $\hat{G}_{7}$ and $\hat{\hat{G}}_{7}$. Alternatively, we can omit together $\tilde{G}_{7}$, it's hatted and doublehatted counterparts, and directly write down the expressions for  $\hat{G}_{7}$ and $\hat{\hat{G}}_{7}$ and show them to satisfy the required equations, as checked explicitly here.

\end{remark}

\subsection{In Twistorial Cohomotopy}

 Now, Hypothesis H {\it generalizes} to heterotic M-Theory where Cohomotopy is enhanced to twistorial Cohomotopy (constructed in \cite[\S 2.45]{FSS23-Char}, \cite[\S 2.14]{FSS22-GS}). Here the target $\infty$-stack is represented by $\mathbb{C}P^{3}\sslash\mathrm{Sp}(2)$ instead of $S^{4}\sslash\mathrm{Sp}(2)$, and we have the Borel equivariantized twistor fibration $t_{\mathbb{H}}\sslash\mathrm{Sp}(2):\mathbb{C}P^{3}\sslash\mathrm{Sp}(2)\longrightarrow{S^{4}}\sslash\mathrm{Sp}(2)$.
 
Following the above, we define {\it twistorial Cohomotopy} 
\begin{equation}
    \zeta^{\tau}(X):=\mathrm{H}^{\tau}(X; \mathbb{C}P^{3})
    \end{equation}
    where we have the {\it twist} $\tau:X\longrightarrow{\mathrm{BSp}}(2)$.

\begin{definition}[\bf Flux densities]
In the case of twistorial cohomotopy, we are concerned with {\it flux densities} of the following form (\cite[\S 2.14]{FSS22-GS}):
\begin{equation}
\renewcommand{\arraystretch}{1.3}
\label{Twistorial flux densities}
	\Omega_{d R}^{1}\left(-; \mathfrak{l}_{S^{4}\sslash \mathrm{Sp}{(2)}} \mathrm{\mathbb{C}P}^{3}\sslash \mathrm{Sp}(2)\right)_{\mathrm{clsd}}:=\left\{\begin{array}{l|l}
		F_{2} & \dd F_{2}=0 \\H_{3} & \dd H_{3}=\tilde{G}_{4}-\tfrac{1}{2} p_{1}(\omega)-F_{2}^{2} \\
		G_{7} & \dd G_{7}=\tfrac{1}{2} \tilde{G}_{4}\left(\tilde{G}_{4}-\tfrac{1}{2} p_{1}(\omega)\right) \\
		\tilde{G}_{4} & \dd \tilde{G}_{4}=0 \\
		\tfrac{1}{2} p_{1}(\omega) & \dd\left(\tfrac{1}{2} p_{1}(\omega)\right)=0
	\end{array}\right\}
\end{equation}
\end{definition}

\begin{definition}[\bf Gauge potentials and gauge transformations]
Given the flux densities \eqref{Twistorial flux densities} on some $U_{i}$, we have the following traditional choice of {\it gauge potentials}
\begin{equation}
\renewcommand{\arraystretch}{1.3}
\label{Twistorial gauge potentials}
	\left\{\begin{array}{l|l}
		C_{3} \in \Omega_{dR}^{3}\left(U_{i}\right) & \dd C_{3}=\tilde{G}_{4} \\
		C_{6} \in \Omega_{dR}^{6}\left(U_{i}\right) & \dd C_{6}=G_{7}-\tfrac{1}{2} C_{3}\left(\tilde{G}_{4}-\tfrac{1}{2} p_{1}(\omega)\right) \\
		B_{2} \in \Omega_{dR}^{2}\left(U_{i}\right) & \dd B_{2}=H_{3}-C_{3}+ \mathrm{CS}(\omega) + A_{1} F_{2}\\
        A_{1} \in \Omega_{dR}^{1}\left(U_{i}\right) & \dd A_{1}= F_{2}\\
    \end{array}\right\}
\end{equation}
and the corresponding {\it gauge transformations} between the gauge potentials
\begin{equation}
\label{Twistorial gauge tranfomrations}
\renewcommand{\arraystretch}{1.3}
\begin{aligned}
&\begin{array}{c}
\left(C_{3}, C_{6}, B_{2}, A_{1}\right) \sim \left(C_{3}^{\prime}, C_{6}^{\prime}, B_{2}^{\prime}, A_{1}^{\prime}\right) \\
\phantom{\left(C_{3}, C_{6}, B_{2}, A_{1}\right)}\vcenter{\hbox{$\Updownarrow$}}\phantom{\left(C_{3}^{\prime}, C_{6}^{\prime}, B_{2}^{\prime}, A_{1}^{\prime}\right)}
\end{array}
\\[1ex]
&\exists \left\{
\begin{array}{l|l}
C_{2} \in \Omega_{dR}^{2}\left(U_{i}\right) & \dd C_{2} = C_{3}^{\prime} - C_{3} \\
C_{5} \in \Omega_{dR}^{5}\left(U_{i}\right) & \dd C_{5} = C_{6}^{\prime} - C_{6} - \tfrac{1}{2} C_{3}^{\prime} C_{3} - \tfrac{1}{4}C_{2}{p_{1}}(\omega) \\
B_{1} \in \Omega_{dR}^{1}\left(U_{i}\right) & \dd B_{1} = B_{2}^{\prime} - B_{2} + C_{2} - \int dt\, \mathrm{Tr}(\omega^{\prime} - \omega) \wedge \omega_{t} - A_{0} F_{2} \\
A_{0} \in \Omega_{dR}^{0}\left(U_{i}\right) & \dd A_{0} = A_{1}^{\prime} - A_{1}
\end{array}
\right\}.
\end{aligned}
\end{equation}

\end{definition}
\begin{definition}[\bf Null concordances]
Given the flux densities in \eqref{Twistorial flux densities}, we say that a {\it null concordance} is
\begin{equation}
	\begin{array}{c}
		\left(\hat{\tilde{{G}}}_{4}, \hat{G}_{7}, \hat{H}_{3}, \hat{F}_{2}, \hat{p}_{1}(\omega)\right) \in \Omega_{\text {dR}}^{1}\left(U_{i} \times[0,1], \mathfrak{l}_{S^{4}\sslash \mathrm{Sp}(2)} \mathrm{CP}^{3}\sslash \mathrm{Sp}(2)\right)_{ \mathrm{clsd} }
        \end{array}
        \end{equation}
        such that
        \begin{equation}
        \label{Twistorial Boundary conditions}
        \renewcommand{\arraystretch}{1.3}
        \begin{array}{l}
		\iota_{0}^{*}\left(\hat{\tilde{G}}_{4}, \hat{G}_{7}, \hat{H}_{3}, \hat{F}_{2}, \hat{p}_{1}(\omega)\right)=0 \\
		\iota_{1}{ }^{*}\left(\hat{\tilde{G}}_{4}, \hat{G}_{7}, \hat{H}_{3}, \hat{F}_{2}, \hat{p}_{1}(\omega)\right)=\left(\tilde{G}_{4}, G_{7}, H_{3}, F_{2}, p_{1}(\omega)\right)
	\end{array}
\end{equation}
and a {\it null} concordance of concordances between a pair of them is given by
\begin{equation}
	\begin{array}{l}
			\left(\hat{\hat{\tilde{G}}}_{4}, {\hat{\hat{G}}}_{7}, {\hat{\hat{H}}}_{3}, {\hat{\hat{F}}}_{2}, \hat{\hat{p}}_{1}(\omega, \omega')\right) \in \Omega_{dR}^{1}(U_{i} \times \underbrace{[0,1]}_{t} \times \underbrace{[0,1]}_{s} ; \mathfrak{l}_{S^{4}\sslash \mathrm{Sp}(2)}\mathrm{CP}^{3}\sslash \mathrm{Sp}(2))_{\mathrm{ clsd} }
            \end{array}
            \end{equation}
            such that
            \begin{equation}
            \label{Twistorial higher boundary conditions}
            \renewcommand{\arraystretch}{2.2}
            \begin{array}{l}
\iota_{s=0}^{*}\left(\hat{\hat{\tilde{G}}}_{4}, {\hat{\hat{G}}}_{7}, {\hat{\hat{H}}}_{3}, {\hat{\hat{F}}}_{2}, \hat{\hat{p}}_{1}\right)=\left(\hat{\tilde{G}}_{4}, \hat{G}_{7}, \hat{H}_{3}, \hat{F}_{2}, \hat{p}_{1}\right) \\
		\iota^{*}_{s=1}\left(\hat{\hat{\tilde{G}}}_{4}, {\hat{\hat{G}}}_{7}, {\hat{\hat{H}}}_{3}, {\hat{\hat{F}}}_{2}, \hat{\hat{p}}_{1}\right)=\left(\hat{\tilde{G}}_{4}^{\prime}, \hat{G}_{7}^{\prime}, \hat{H}_{3}^{\prime}, \hat{F}_{2}^{\prime},\hat{p}_{1}^{\prime}\right) \\
			\iota^{*}_{t=0}\left(\hat{\hat{\tilde{G}}}_{4}, {\hat{\hat{G}}}_{7}, {\hat{\hat{H}}}_{3}, {\hat{\hat{F}}}_{2}, \hat{\hat{p}}_{1}\right)=0 \\
			{ \iota}^{*}_{t=1}\left(\hat{\hat{\tilde{G}}}_{4}, {\hat{\hat{G}}}_{7}, {\hat{\hat{H}}}_{3}, {\hat{\hat{F}}}_{2}, \hat{\hat{p}}_{1}\right) = \operatorname{pr}_{U_{i}}^{*}\left(\tilde{G}_{4}, G_{7}, H_{3}, F_{2}, p_{1}\right)
		\end{array}\\
    \end{equation}
	where $$pr_{U_{i}}: U_{i} \times[0,1]_{S} \longrightarrow U_{i}$$
    \end{definition}

Similarly to the last section:
    
\begin{proposition}[\bf Gauge potentials from null concordances]
Every traditional gauge potential in twistorial cohomotopy arises from the null concordances in the following way:
\begin{equation}
\label{Twistorial gauge potentials from null concordances}
\renewcommand{\arraystretch}{1.5}
	\begin{array}{l}
		
		\begin{array}{l}
			2\mathrm{CS}(\omega)= s_{3}=\int_{[0,1]}\hat{p}_{1}\\

			C_{3}=\int_{[0,1]} \hat{G}_{4}+\tfrac{1}{4} \int_{[0,1]} \hat{p}_{1}=\int_{[0,1]} \hat{\tilde{G}}_{4}.\\
			C_{6}=\int_{[0,1]} \left(\hat{G}_{7}-\tfrac{1}{2}\left(\int_{[0,-]}\hat{\tilde{G}}_{4}\right)\left(\hat{\tilde{G}}_{4}-\tfrac{1}{2} \hat{p}_{1}\right)\right) .\\
			B_{2}=\int_{[0,1]} \hat{H}_{3} \\
            A_{1}=\int_{[0,1]} \hat{F}_{2}\,.
		\end{array}
	\end{array}
\end{equation}
\end{proposition}
\begin{proof}
We show the above expressions for the gauge potentials satisfy the expected Bianchi identities as mentioned in \eqref{Twistorial gauge potentials}.

We start with $B_{2}$
\begin{equation}
\renewcommand{\arraystretch}{1.5}
	\begin{array}{ll}
		\dd B_{2}&=\dd \int_{[0,1]} \hat{H}_{3}=\iota^{*}_{1} \hat{H_{3}}-\iota^{*}_{0}\hat{{H}_{3}}-\int_{[0,1]} \dd \hat{H}_{3}=H_{3}-0-\int_{[0,1]}\left(\hat{G}_{4}-\frac{1}{4} \hat{p}_{1}- \hat{F}_{2}^{2}\right)
        \\
        &=H_{3}-0-\int_{[0,1]} \hat{\tilde{G}}_{4}+\frac{1}{2} \int_{[0,1]} \hat{p}_{1}+ \int_{[0,1]}\hat{F}_{2}^{2}
        \\
	&	=H_{3}-C_{3}+\frac{1}{2} \int_{[0.1]} \hat{p}_{1} + A_{1} F_{2}\\
       & =H_{3}-C_{3}+\mathrm{CS}(\omega) + A_{1} F_{2}\,.
	\end{array}
    \end{equation}
    Then, we want to check the $A_{1}$ potential
    \begin{equation}
    \dd A_{1}=\dd \int_{[0,1]} \hat{F}_{2}=\iota^{*}_{1} \hat{F_{2}}-\iota^{*}_{0}\hat{{F}_{2}}-\int_{[0,1]} \dd \hat{F}_{2}=F_{2}-0-0= F_{2}\,.
    \end{equation}
So the above relations satisfy the usual Bianchi identities.
\end{proof}

 \begin{proposition}[\bf Surjections]
    Let us turn to the surjections from the null concordances to the gauge potentials explicitly
    \vspace{-2mm} 
\begin{equation}
\label{Twistorial surjections}
\renewcommand{\arraystretch}{1.4}
	\begin{array}{l}
		\hat{p}_{1}=t{p_{1}}+\dd t\, s_{3} \\
        \hat{\tilde{G}}_{4}=
			t (\left(G_{4}+\tfrac{1}{4}p_{1}\right)+\dd t\, C_{3} = t \tilde{G_{4}}+ \dd t\, C_{3}\\ \hat{G}_{7}=t^{2} G_{7}+2 t \dd t\, C_{6}+\tfrac{1}{4}t \dd t C_{3} s_{3}\\
            \hat{H}_{3}=t H_{3}+\dd t\, B_{2}   + t(1-t) A_{1} F_{2}\\
		\hat{F}_{2}=t F_{2}+ \dd t \,A_{1}\,.	
		\end{array}\\
        \end{equation}
\end{proposition}

\begin{proof}
We immediately see that they are consistent when plugged in \eqref{Twistorial gauge potentials from null concordances} and they satisfy the usual null concordance conditions at $t=0,1$.

We want to check whether these relations satisfy the Bianchi identities in \eqref{Twistorial gauge potentials}.
We start with $\hat{H}_{3}$.
        \begin{equation}
        \hspace{-5mm}
        \renewcommand{\arraystretch}{1.4}
        \begin{array}{ll}
\dd \hat{H}_{3}&=\dd t\, H_{3}+t\left(G_{4}-\tfrac{1}{4} p_{1}- F_{2}^{2}\right)-\dd t\left(H_{3}-C_{3}+\tfrac{1}{2}s_{3}+ A_{1} F_{2}\right) + t(1-t) F_{2}^{2} + (1-2t)\dd t A_{1} F_{2}\\
&=t\left(G_{4}-\tfrac{1}{4} p_{1}- F_{2}^{2}\right)+\dd t\left(C_{3}-\tfrac{1}{2}s_{3}- A_{1} F_{2}\right)  -2t \dd t A_{1} F_{2} -t^{2} F_{2}^{2}\\ 
&= t\left(G_{4}-\tfrac{1}{4} p_{1}\right)+\dd t\left(C_{3}-\tfrac{1}{2}s_{3}\right) -2t \dd t  A_{1} F_{2}  -t^{2} F_{2}^{2}
\\ 
&= \hat{\tilde{G}}_{4}- \tfrac{1}{2}\hat{p}_{1}-\hat{F}_{2}^{2}\,.
\end{array}
\end{equation}
Also, we find $$\hat{F}_{2}^{2}= (tF_{1}+ \dd t\, A_{1})(tF_{1}+ \dd t\, A_{1})= 2t\dd t\, A_{1} F_{2}+ t^{2} F_{2}^{2}\,.
$$
We now want to check for $\hat{F}_{2}$
\begin{equation}
\dd\hat{F}_{2}= \dd(tF_{2}+\dd t\, A_{1})= \dd t\, F_{2}+0 - \dd t\, F_{2}=0.
\end{equation}
So the relations are consistent.
\end{proof}

Let us turn to the concordance of concordances.

\begin{proposition}[\bf Gauge equivalences from null-concordance of concordances]
In twistorial cohomotopy, the null concordances of concordances surject onto usual gauge transformations between the usual gauge potentials
\begin{equation}
\label{Twistorial gauge equivalnes from null concordances of concordances}
\renewcommand{\arraystretch}{1.5}
\begin{array}{l}
	C_{2}=\int_{s \in{[0,1]}} \int_{t \in[0,1)} \hat{\hat{\tilde{G}}}_{4}\\ B_{1}=\int_{s \in{[0,1]}} \int_{t \in{[0,1]}} {\hat{\hat{H}}}_{3}\\C_{5}=\int_{s \in{[0,1]}}\int_{ t \in{[0,1]}} (\hat{\hat{{G}}}_{7}-\tfrac{1}{2}\hat{\hat{C}}_{3}(\hat{\hat{\tilde{G}}}_{4}-\tfrac{1}{2}\hat{\hat{p}}_{1}))-\tfrac{1}{2} C_{2} C_{3}\\ A_{0}= \int_{s \in{[0,1]}} \int_{t \in[0,1)} \hat{\hat{F}}_{2}\\
    s_{2}= \int_{s \in{[0,1]}} \int_{t \in[0,1)} \hat{\hat{p}}_{1}\,.
    \end{array}
\end{equation}
\end{proposition}
\begin{proof}
We want to check whether these relations satisfy the Bianchi identities in \eqref{Twistorial gauge tranfomrations}.

Firstly, we want to check $B_{1}$.
 \begin{equation}
\renewcommand{\arraystretch}{1.8}
  \begin{array} {l}  
 \dd B_{1}=\dd \int_{s\in{[0,1]}} \int_{t\in{[0,1]}}{\hat{\hat{H}}}_{3}=\iota^{*}_{s=1}\int_{t \in{[0,1]}} {\hat{\hat{H}}}_{3}-\iota^{*}_{s=0} \int_{t \in{[0,1]}}{\hat{\hat{H}}}_{3}-\int_{s \in{[0,1]}} \dd \int_{t \in{[0,1]}} {\hat{\hat{H}}}_{3} \\

	=\int_{t\in[0,1]} \iota^{*}_{s=1} \hat{\hat{H}}_{3}-\int_{t \in[0,1]} \iota^{*}_{s=0} {\hat{\hat{H}}}_{3}-\int_{s \in[0,1]}(\iota^{*}_{t=1} \hat{\hat{H}}_{3}-\iota^{*}_{t=0} \hat{\hat{H}}_{3})+\int_{s \in[0,1]}\int_{ t\in{[0,1]}} \dd {\hat{\hat{H}}}_{3} \\
    =\int_{t\in[0,1]}\hat{H^{\prime}}_{3}- \int_{t\in[0,1]}\hat{H}_{3}-0+\int_{s\in[0,1]}\int_{t\in[0,1]}\hat{\hat{\tilde{G}}}_{4}-\tfrac{1}{2}\int_{s\in[0,1]}\int_{t\in[0,1]}\hat{\hat{p}}_{1}-\int_{s\in[0,1]}\int_{t\in[0,1]}\hat{\hat{F}}_{2}^{2} \\
    =B_{2}^{\prime}-B_{2}+C_{2}-\tfrac{1}{2}\int_{s\in[0,1]}\int_{t\in[0,1]}\hat{\hat{p}}_{1}- A_{0}F_{2}\\
    =B_{2}^{\prime}-B_{2}+C_{2}-\tfrac{1}{2}s_{2}- A_{0}F_{2}\,.
    \end{array}
    \end{equation}
    Next, we want to check $A_{0}$.
    \begin{equation}
\renewcommand{\arraystretch}{1.7}
    \begin{array}{ll}
     \dd A_{0}=\dd \int_{s\in{[0,1]}} \int_{t\in{[0,1]}}{\hat{\hat{F}}}_{2}=\iota^{*}_{s=1}\int_{t \in{[0,1]}} {\hat{\hat{F}}}_{2}-\iota^{*}_{s=0} \int_{t \in{[0,1]}}{\hat{\hat{F}}}_{2}-\int_{s \in{[0,1]}} \dd \int_{t \in{[0,1]}} {\hat{\hat{F}}}_{2}\\
     =\int_{t\in[0,1]} \iota^{*}_{s=1} \hat{\hat{F}}_{2}-\int_{t \in[0,1]} \iota^{*}_{s=0} \hat{\hat{F}}_{2}-\int_{s \in[0,1]}\left(\iota^{*}_{t=1} \hat{\hat{F}}_{2}-\iota^{*}_{t=0} \hat{\hat{F}}_{2}\right)+\int_{s \in[0,1]}\int_{ t\in[0,1]} \dd {\hat{\hat{F}}}_{2} \\
    =\int_{t\in[0,1]}\hat{F^{\prime}}_{2}- \int_{t\in[0,1]}\hat{F_{2}}-\int_{s\in{[0,1]}}\hat{F}_{2}+ \int_{s\in{[0,1]}} \int_{t\in{[0,1]}}{\dd\hat{\hat{F}}}_{2}\\
    =A_{1}^{\prime}-A_{1}\,.
    \end{array}
    \end{equation}
    So they are consistent with our expectations.
    \end{proof}
    
    Let us now turn to the surjections from the  null concordance of concordances to the gauge transformations between the gauge potentials in twistorial cohomotopy explicitly.
    
    \begin{proposition}[\bf Surjections] The surjections are given by setting:
    \begin{equation}
    \label{Twistorial higher surjections}
\renewcommand{\arraystretch}{1.5}
    \begin{array}{ll}
    \hat{\hat{p}}_{1}&=t p_{1}+\dd t s_{3}+s\dd t\left(s_{3}^{\prime}-s_{3}\right)-\dd s\, \dd t\, s_{2},
    \\
        {\hat{\hat{\tilde{G}}}}_{4} &=t{\tilde{G}_{4}}+\dd t \,C_{3}+s \dd t\left(C_{3}^{\prime}-C_{3}\right)-\dd s\,\dd t\, C_{2},
        \\
        {\hat{\hat{G}}}_{7}
		&= t^{2} G_{7}+2 t \dd t\, C_{6}+ \tfrac{1}{4} t\dd t\, C_{3}s_{3}+ 2 s\, t\, \dd t \left(C_{6}^{\prime}-C_{6}\right)-2 t\, \dd s\, \dd t\left(C_{5}+\tfrac{1}{2} C_{2}, C_{3}\right) 
        \\
		&-  \tfrac{1}{4}s t\,\,\dd t C_{3}\left(s_{3}-s_{3}^{\prime}\right)+\tfrac{1}{4}s \,t\, \dd t\left(C_{3}^{\prime}-C_{3}\right) s_{3}-\tfrac{1}{4} s^{2} t \dd t\left(C_{3}^{\prime}-C_{3}\right)\left(s_{3}-s_{3}^{\prime}\right) 
        \\
		& +\tfrac{1}{4}t\, \dd s\, \dd t(\left(C_{3} s_{2}-C_{2} s_{3}\right)+\tfrac{1}{4} st \dd s\, \dd t\left(C_{3}^{\prime}-C_{3}\right) s_{2}-\tfrac{1}{4}s t\, \dd s\, \dd t C_{2}\left(s_{3}^{\prime}-s_{3}\right),
         \\[4pt]
    \hat{\hat{H}}_{3}&= t H_{3}+ \dd t\, B_{2} + t(1-t) A_{1} F_{2}
    + s\dd t\, (B_{2}^{\prime}-B_{2})\\
    &+ s t (1-t) (A_{1}^{\prime}- A_{1}) F_{2} + \dd s\, (t-t^2) A_{0} F_{2}- \dd s\, \dd t\, B_{1},
    \\[3pt]
   \hat{\hat{F}}_{2}&= tF_{2}+ \dd t A_{1} + s\dd t (A_{1}^{\prime}-A_{1})- \dd s\, \dd t\, A_{0}\,.
   \end{array}
    \end{equation}
    \end{proposition}

    \begin{proof}
    We can see that \eqref{Twistorial gauge equivalnes from null concordances of concordances} is satisfied and also the usual conditions for null concordance of concordances at $s=0,1$ and $t=0,1$.
    
    Similarly, we can compute 
    \begin{equation}
        \hat{\hat{F}}_{2}^{2}=t^{2} F_{2}^{2} + 2t \dd t A_{1} F_{2}+ 2 st \dd t (A_{1}^{\prime}-A_{1})F_{2}- 2 t \dd s\, \dd tA_{0} F_{2}\,.
\end{equation}
Let us check these equations to be consistent with \eqref{Twistorial gauge tranfomrations}.
We want to start with
\begin{equation}
\begin{array}{l}
\dd\hat{\hat{F}}_{2}= \dd t F_{2} -\dd t F_{2} + \dd s\, \dd t (A_{1}^{\prime}-A_{1})+0-\dd s\, \dd t (A_{1}^{\prime}-A_{1})=0.
\end{array}
\end{equation}
    Now, we want to see 
\begin{equation}
\renewcommand{\arraystretch}{1.4}
\begin{array}{l}
   \dd\hat{\hat{H}}_{3}= \dd t\, H_{3}+ t (G_{4}-\tfrac{1}{4}p_{1}- F_{2}^{2}) -\dd t\, (H_{3}-C_{3}+\tfrac{1}{2} s_{3} + A_{1}F_{2}) + t(1-t) F_{2}^{2} + (1-2t)\dd \,t A_{1}F_{2}\\- \dd s\, \dd t (B_{2}^{\prime}-B_{2}+C_{2}-\tfrac{1}{2}s_{2}-A_{0}F_{2})+ \dd s\, \dd t \,(B_{2}^{\prime}-B_{2}) -s \dd t (-C_{3}^{\prime}+\tfrac{1}{2}s_{3}^{\prime}+A_{1}^{\prime}F_{2}+C_{3} - \tfrac{1}{2}s_{3}- A_{1}F_{2})\\ + \dd s t(1-t) (A_{1}^{\prime}-A_{1})F_{2}+ s(1-2t)\dd t\, (A_{1}^{\prime}-A_{1})F_{2} -ds (1-2t)\dd t\, A_{0}F_{2}- \dd s\, t(1-t) (A_{1}^{\prime}-A_{1})F_{2}
   \\[5pt]
    = (t \tilde{G}_{4}+\dd t\, C_{3}+s \dd t\,\left(C_{3}^{\prime}-C_{3}\right)-\dd s \dd\, t C_{2} )- \tfrac{1}{2}(t p_{1}+\dd t s_{3}+s \dd t\,\left(s_{3}^{\prime}-s_{3}\right)-\dd s \dd t \,s_{2})\\- (t^{2} F_{2}^{2} + 2t \dd t\, A_{1} F_{2}+ 2 st\, \dd t\, (A_{1}^{\prime}-A_{1})F_{2}- 2 t \dd s\, \dd t\,A_{0} F_{2})
    \\[5pt]
    =\hat{\hat{\tilde{G}}}_{4}-\tfrac{1}{2}\hat{\hat{p}}_{1}
-\hat{\hat{F}}_{2}^{2}
\end{array}
    \end{equation}
as expected from the Bianchi identities.
\end{proof}
\begin{remark}
In this section on twistorial cohomotopy, we have only proved the results that new or are different from the tangentially twisted case; the ones not proven here follow the same results from the last section on tangentially twisted gauge potentials.
\end{remark}
\begin{remark}
The Borel-equivariantized quaternionic Hopf fibration $$h_{\mathbb{H}}\sslash \mathrm{Sp}(2):S^{7}\sslash\mathrm{Sp}(2)\longrightarrow{S^{4}}\sslash\mathrm{Sp}(2)$$ factorizes through the 7d Borel-equivariantized complex Hopf fibration $$h_{\mathbb{C}}\sslash\mathrm{Sp}(2):S^{7}\sslash\mathrm{Sp}(2)\rightarrow{\mathbb{C}P^{3}}\sslash\mathrm{Sp}(2)$$ and the Borel-equivariantized twistor fibration 
$$
t_{\mathbb{H}}\sslash\mathrm{Sp}(2):{\mathbb{C}P^{3}}\sslash\mathrm{Sp}(2)\longrightarrow{S^{4}}\sslash\mathrm{Sp}(2).
$$
    If we add the 7d Borel-equivariantized complex Hopf fibration $$h_{\mathbb{C}}\sslash\mathrm{Sp}(2):S^{7}\sslash\mathrm{Sp}(2)\rightarrow{\mathbb{C}P^{3}}\sslash\mathrm{Sp}(2)$$
    into the story then we have one more generator (on top of the twistorial cohomotopy fluxes) corresponding to the flux \cite[\S 2.14]{FSS22-GS}
    \begin{equation}
   \begin{array}{l}
       ({h}_{1},...)\in \Omega^{1}_{\mathrm{dR}}(-;\mathfrak{l}_{\mathbb{C}P^{3}\sslash\mathrm{Sp(2)}_{S^{4}\sslash\mathrm{Sp(2)}}}(S^{7}\sslash\mathrm{Sp(2)))}_{\mathrm{clsd}}
   \end{array}    
   \end{equation}
   where the dots denote the fluxes already mentioned in \eqref{Twistorial flux densities}.
   
   Thus on top of the other fluxes, we have
   
    \begin{equation}
    \label{Extra Flux}
        h_{1}\in \Omega^{1}_{\mathrm{dR}}(U_{i}), \hspace{0.2cm} \dd h_{1}=F_{2}
    \end{equation}
    and the corresponding gauge potential (on top of \eqref{Twistorial gauge potentials}) to be 
    \begin{equation}
    \label{Extra gauge potential}
\phi\in\Omega^{0}_{\mathrm{dR}}(U_{i}), \hspace{0.2cm} \dd\phi= h_{1}-A_{1}\,.
    \end{equation}
   The corresponding null concordance is
   \begin{equation}
   \begin{array}{l}
       (\hat{h}_{1},...)\in \Omega^{1}_{\mathrm{dR}}({U_{i}\times[0,1];\mathfrak{l}_{\mathbb{C}P^{3}\sslash\mathrm{Sp(2)}_{S^{4}\sslash\mathrm{Sp(2)}}}(S^{7}\sslash\mathrm{Sp(2))}})_{\mathrm{clsd}};\qquad
       \iota^{*}_{0}\hat{h}_{1}= 0, \hspace{0.2cm}\iota^{*}_{1}\hat{h}_{1}= h_{1}
   \end{array}    
   \end{equation}
   where the dots denote the null concordances already discussed earlier. There is a surjection from the null concordance $\hat{h}_{1}$ to the usual gauge potential $\phi$, given by
   \begin{equation}
       \phi= \int_{t\in[0,1]}\hat{h}_{1}\,.
   \end{equation}
   and we check its consistency.
   \begin{equation}
   \renewcommand{\arraystretch}{1.5}
   \begin{array}{ll}
       \dd\phi&= \dd \int_{[0,1]} \hat{h}_{1}
       \\
       &=\iota^{*}_{1} \hat{h}_{1}-\iota^{*}_{0}\hat{{h}}_{1}-\int_{[0,1]} \dd \hat{h}_{1}=h_{1}-0-\int_{[0,1]}  \hat{F}_{2}= h_{1}- A_{1}\,.
       \end{array}
   \end{equation}
   Thus, it is consistent with the Bianchi identity \eqref{Extra gauge potential}.   We now give the explicit form of the surjection  from the null concordance $\hat{h}_{1}$ to the usual gauge potential $\phi$ to be
   \begin{equation}
       \hat{h}_{1}= t^{2}h_{1}+2t\dd t \, \phi + t(1-t)A_{1}\,.
   \end{equation}
   We can easily check the expected conditions at t=0,1 for this null concordance.

   Also, we can calculate 
   \begin{equation}
      \dd\hat{h}_{1}= t F_{2} + \dd t\, A_{1} = \hat{F}_{2}
   \end{equation}
   which is the expected Bianchi identity \eqref{Extra Flux}.

   For the null-concordance of concordances, it's trivial for the $h_{1}$ flux, as one can convince oneself easily due to the absence of a `-1-form'.

\end{remark}

\section{Gauge potentials on M5-probes of 11D SuGra in Equivariant Twistorial cohomotopy}
Now, we want to take a step further and place the M5 brane on an orbifold. We will be interested in the case of $\mathrm{Sp}(1)$ parametrized $\mathbb{Z}_{2}$-equivariant twistorial cohomotopy, as these two commute with each other as subgroups of $\mathrm{Sp}(2)$ \footnote { Since $\mathbb{Z}_{2}$ and $\mathrm{Sp}(1)$ commute inside $\mathrm{Sp}(2)$, the homotopy quotient $(\mathrm{CP}^{3}\sslash \mathrm{Sp}(1))$ still admits the structure of a $G$-space for $G= \mathbb{Z}_{2}$, fibered over $\mathrm{BSp}(1)$. For the case of $\mathbb{Z}_{2}$-equivariant twistorial cohomotopy, the non-trivial fixed locus is $(\mathrm{CP}^{3}\sslash \mathrm{Sp}(1))^{\mathbb{Z}_{2}}\simeq S^{2}\sslash \mathrm{Sp}(1)$, apart from the obvious $\mathrm{CP}^{3}\sslash\mathrm{Sp}(1)$. Here, we care about the non-abelian character map for the case of classifying spaces that are equivariantly (namely, fixed locus wise) simply connected. If we drop the assumption of stagewise simple-connectedness, things get more involved.} The particular choice is discussed in \cite[p.7]{SS20-Equitwistorial}.

\begin{definition}[\bf Flux densities] In the case of equivariant twistorial cohomotopy \cite[\S 2.48]{SS20-Equitwistorial}, we are concerned with {\it flux densities} of the following form (\cite[\S 1.1]{SS20-Equitwistorial}):
\begin{equation} 
\begin{aligned}
\label{Equitwistorial flux densities}
\hspace{-3mm}
	\Omega_{\rm d R}^{1}
    \left(\! \orbisingular X\sslash{\mathbb{Z}_{2}}\,; \; \mathfrak{l}_{\orbisingular} (S^{4}\sslash{\mathbb{Z}_{2})\sslash \mathrm{Sp}{(1)}} 
    \orbisingular (\mathrm{\mathbb{C}P}^{3}\sslash{\mathbb{Z}}_{2})\sslash \mathrm{Sp}(1)\right)_{\mathrm{clsd}}
   & :=\left\{\!
    \begin{array}{l|l}
		F_{2} & \dd F_{2}=0  \\H_{3} & \dd H_{3}=\tilde{G}_{4}-\tfrac{1}{2} p_{1}(\omega)-F_{2}^{2} \\
		G_{7} & \dd G_{7}=\tfrac{1}{2} \tilde{G}_{4}\left(\tilde{G}_{4}-\tfrac{1}{2} p_{1}(\omega)\right) \\
		\tilde{G_{4}} & \dd \tilde{G_{4}}=0 \\
        \tfrac{1}{4} p_{1}(\omega) & \dd \left(\tfrac{1}{4} p_{1}(\omega)\right)=0\\
	\end{array}
    \!\!\!\right\}
    \\
   & \rightarrow
    \left\{\begin{array}{l|l}
		F_{2} & \dd F_{2}|_{X^{\mathbb{Z}_{2}}}=0  \\H_{3} & \dd H_{3}|_{X^{\mathbb{Z}_{2}}}=-\frac{1}{2} p_{1}(\omega)-F_{2}^{2} \\
		G_{7} &  G_{7}|_{X^{\mathbb{Z}_{2}}}=0 \\
		\tilde{G_{4}} &  \tilde{G_{4}}|_{X^{\mathbb{Z}_{2}}}=0 \\
        \tfrac{1}{4} p_{1}(\omega) & \dd \left(\tfrac{1}{4} p_{1}(\omega)\right)|_{X^{\mathbb{Z}_{2}}}= 0\\
	\end{array}
    \right\}
    \end{aligned}
    \end{equation}

    \medskip 
\noindent where the fluxes in the bulk $\mathrm{X}^{11}$ correspond to the first blob, and the blob below it corresponds to the fluxes along the $\mathbb{Z}_{2}$-fixed locus $\mathrm{X}^{\mathbb{Z}_{2}}$.
    In what follows below, we will denote open charts on the bulk $\mathrm{X}^{11}$ by $\mathrm{U}_{i}$, and open charts on the fixed locus $\mathrm{X}^{\mathbb{Z}_{2}}$ by $\mathrm{U}_{i}^{\mathbb{Z}_{2}}$.
\end{definition}

\begin{definition}[\bf Gauge potentials and gauge transformations]
Given the flux densities \eqref{Equitwistorial flux densities} on some $U_{i}$, we have the following traditional choice of {\it gauge potentials}
\begin{equation}
\begin{aligned}
\renewcommand{\arraystretch}{1.3}
\label{Equitwistorial gauge potentials}
	\left\{\begin{array}{l|l}
		C_{3} \in \Omega_{dR}^{3}\left(U_{i}\right) & \dd C_{3}=\tilde{G}_{4} \\
		C_{6} \in \Omega_{dR}^{6}\left(U_{i}\right) & \dd C_{6}=G_{7}-\tfrac{1}{2} C_{3}\left(\tilde{G}_{4}-\tfrac{1}{2} p_{1}(\omega)\right) \\
		B_{2} \in \Omega_{dR}^{2}\left(U_{i}\right) & \dd B_{2}=H_{3}-C_{3}+ \mathrm{CS}(\omega) + A_{1} F_{2}\\
        A_{1} \in \Omega_{dR}^{1}\left(U_{i}\right) & \dd A_{1}= F_{2}\\
    \end{array}\right\}
    \\ \rightarrow
    \left\{\begin{array}{l|l}
		
		B_{2} \in \Omega_{dR}^{2}\left(U_{i}^{\mathbb{Z}_{2}}\right) & \dd B_{2}|_{U_{i}^{\mathbb{Z}_{2}}}=H_{3}+ \mathrm{CS}(\omega) + A_{1} F_{2}\\
        A_{1} \in \Omega_{dR}^{1}\left(U_{i}^{\mathbb{Z}_{2}}\right) & \dd A_{1}|_{U_{i}^{\mathbb{Z}_{2}}}= F_{2}
    \end{array}\right\}
    \end{aligned}
\end{equation}
and the corresponding {\it gauge transformations} between the gauge potentials $\left(C_{3}, C_{6}, B_{2}, A_{1}\right)$ and $ \left(C_{3}^{\prime}, C_{6}^{\prime}, B_{2}^{\prime},A_{1}^{\prime}\right)$
\begin{equation}
\label{Equitwistorial gauge tranfomrations}
\renewcommand{\arraystretch}{1.0}
\begin{aligned}
\left\{\begin{array}{l|l}
C_{2} \in \Omega_{dR}^{2}\left(U_{i}\right) & \dd C_{2} = C_{3}^{\prime} - C_{3} \\
C_{5} \in \Omega_{dR}^{5}\left(U_{i}\right) & \dd C_{5} = C_{6}^{\prime} - C_{6} - \tfrac{1}{2} C_{3}^{\prime} C_{3} - \tfrac{1}{4}C_{2} {p_{1}}(\omega) \\
B_{1} \in \Omega_{dR}^{1}\left(U_{i}\right) & \dd B_{1} = B_{2}^{\prime} - B_{2} + C_{2} - \int dt\, \mathrm{Tr}(\omega^{\prime} - \omega) \wedge \omega_{t} - A_{0} F_{2} \\
A_{0} \in \Omega_{dR}^{0}\left(U_{i}\right) & \dd A_{0} = A_{1}^{\prime} - A_{1}
\end{array}\right\}
\\ 
\rightarrow
\left\{\!\!
\begin{array}{l|l}

B_{1} \in \Omega_{dR}^{1}\left(U_{i}^{\mathbb{Z}_{2}}\right) & \dd B_{1}|_{U_{i}^{\mathbb{Z}_{2}}} = B_{2}^{\prime} - B_{2}  - \int dt\, \mathrm{Tr}(\omega^{\prime} - \omega) \wedge \omega_{t} - A_{0} F_{2} \\
A_{0} \in \Omega_{dR}^{0}\left(U_{i}^{\mathbb{Z}_{2}}\right) & \dd A_{0}|_{U_{i}^{\mathbb{Z}_{2}}} = A_{1}^{\prime} - A_{1}
\end{array}
\right\}.
\end{aligned}
\end{equation}
{\it The gauge potentials corresponding to the fluxes that are not supported at the fixed locus $\mathrm{X}^{\mathbb{Z}_{2}}$ get decoupled at the orbi-fixed locus.}
\end{definition}
\begin{definition}[\bf Null concordances]
Given the flux densities in \eqref{Equitwistorial flux densities}, we say that a {\it null concordance} is
\begin{equation}
	\begin{array}{c}
		\left(\hat{\tilde{G}}_{4}, \hat{G}_{7}, \hat{H}_{3}, \hat{F}_{2}, \hat{p}_{1}\right) \in \Omega_{\text {dR}}^{1}\left(\orbisingular U_{i}\sslash{\mathbb{Z}}_{2} \times[0,1], \mathfrak{l}_{\orbisingular(S^{4}\sslash{\mathbb{Z}}_{2})\sslash \mathrm{Sp}(1)} \orbisingular(\mathrm{CP}^{3}\sslash{\mathbb{Z}}_{2})\sslash \mathrm{Sp}(1)\right)_{ \mathrm{clsd} }
        \end{array}
        \end{equation}
        such that
        \begin{equation}
        \label{Equitwistorial boundary}
        \renewcommand{\arraystretch}{1.3}
        \begin{array}{l}
		\iota_{0}^{*}\left(\hat{\tilde{G}}_{4}, \hat{G}_{7}, \hat{H}_{3}, \hat{F}_{2}, \hat{p}_{1}\right)|_{U_{i}^{\mathbb{Z}_{2}}}=0 \\
		\iota_{1}{ }^{*}\left(\hat{\tilde{G}}_{4}, \hat{G}_{7}, \hat{H}_{3}, \hat{F}_{2}, \hat{p}_{1}\right)|_{U_{i}^{\mathbb{Z}_{2}}}=\left(0, 0, H_{3}, F_{2}, p_{1}\right)
	\end{array}
\end{equation}
and a {\it null} concordance of concordances between a pair of them is given by
\begin{equation}
	\begin{array}{l}
			\left(\hat{\hat{\tilde{G}}}_{4}, {\hat{\hat{G}}}_{7}, {\hat{\hat{H}}}_{3}, {\hat{\hat{F}}}_{2}, \hat{\hat{p}}_{1}\right) \in \Omega_{dR}^{1}(\orbisingular U_{i} \sslash{\mathbb{Z}}_{2}\times \underbrace{[0,1]}_{t} \times \underbrace{[0,1]}_{s} ; \mathfrak{l}_{\orbisingular(S^{4}\sslash{\mathbb{Z}}_{2})\sslash \mathrm{Sp}(1)}\orbisingular(\mathrm{CP}^{3}\sslash{\mathbb{Z}}_{2})\sslash \mathrm{Sp}(1))_{\mathrm{ clsd} }
            \end{array}
            \end{equation}
            such that
            \begin{equation}
            \label{Equitwistorial higher boundary}
            \renewcommand{\arraystretch}{2}
            \begin{array}{l}
		\iota_{s=0}^{*}\left(\hat{\hat{\tilde{G}}}_{4}, {\hat{\hat{G}}}_{7}, {\hat{\hat{H}}}_{3}, {\hat{\hat{F}}}_{2}, \hat{\hat{p}}_{1}\right)=\left(0, 0, \hat{H}_{3}, \hat{F}_{2}, \hat{p}_{1}\right) \\
		\iota^{*}_{s=1}\left(\hat{\hat{\tilde{G}}}_{4}, {\hat{\hat{G}}}_{7}, {\hat{\hat{H}}}_{3}, {\hat{\hat{F}}}_{2}, \hat{\hat{p}}_{1}\right)=\left(0, 0, \hat{H}_{3}^{\prime}, \hat{F}_{2}^{\prime}, \hat{p}_{1}^{\prime}\right) \\
			\iota^{*}_{t=0}\left(\hat{\hat{\tilde{G}}}_{4}, {\hat{\hat{G}}}_{7}, {\hat{\hat{H}}}_{3}, {\hat{\hat{F}}}_{2}, \hat{\hat{p}}_{1}\right)=0 \\
			{ \iota}^{*}_{t=1}\left(\hat{\hat{\tilde{G}}}_{4}, {\hat{\hat{G}}}_{7}, {\hat{\hat{H}}}_{3}, {\hat{\hat{F}}}_{2},\hat{\hat{p}}_{1}\right) = \operatorname{pr}_{U_{i}}^{*}\left(0, 0, H_{3}, F_{2}, p_{1}\right)
		\end{array}\\
    \end{equation}
	where we have shown {\it only} the boundary conditions of the concordances and higher concordances on the chart $\mathrm{U_{i}}^{\mathbb{Z}_{2}}$ of the fixed locus $\mathrm{X}^{\mathbb{Z}_{2}}$, the bulk relations are the same as in \eqref{Twistorial Boundary conditions} and \eqref{Twistorial higher boundary conditions}.
    \end{definition}
\begin{proposition}
[\bf Gauge potentials from null concordances]
1) Every traditional gauge potential in equivariant twistorial cohomotopy arises from the null concordances in the following way:
\begin{equation}
\label{Equitwistorial gauge potentials from null concordances}
\renewcommand{\arraystretch}{1.3}
	\begin{array}{l}
		\begin{array}{rl}
		2 \mathrm{CS}(\omega)|_{U_{i}^{\mathbb{Z}_{2}}}
        & 
        = s_{3}=\int_{[0,1]}\hat{p_{1}}(\omega)\\
			B_{2}|_{U_{i}^{\mathbb{Z}_{2}}} &=\int_{[0,1]} \hat{H}_{3} \\
            A_{1}|_{U_{i}^{\mathbb{Z}_{2}}}&=\int_{[0,1]} \hat{F}_{2}\,.
		\end{array}
	\end{array}
\end{equation}
These are the results for the fixed locus, the bulk results are shown in \eqref{Twistorial gauge potentials from null concordances}.

2) The surjections from the null concordances to the usual potentials in the bulk $\mathrm{X}^{11}$ take the same forms as in \eqref{Twistorial surjections} and take the following forms on the fixed locus $\mathrm{X}^{\mathbb{Z}_{2}}$:
\begin{equation}
\label{Equitwistorial surjections}
\renewcommand{\arraystretch}{1.3}
	\begin{array}{l}
    \hat{p}_{1}= t p_{1} + \dd t\, s_{3}\\
            \hat{H}_{3}=t H_{3}+\dd t\, B_{2}   + t(1-t) A_{1} F_{2}\\
		\hat{F}_{2}=t F_{2}+ \dd t \,A_{1}.\\	
		\end{array}\\
        \end{equation}
\end{proposition}
They can be easily seen to satisfy the conditions in \eqref{Equitwistorial boundary} and also be consistent with \eqref{Equitwistorial gauge potentials}. The proofs for the fixed loci $\mathrm{X}^{\mathbb{Z}_{2}}$ follow analogous steps as in section 3.2 by setting $\tilde{G}_{4}|_{\mathrm{X}^{\mathbb{Z}_{2}}}=0,  {G}_{7}|_{\mathrm{X}^{\mathbb{Z}_{2}}}=0  $.

Let us turn to the concordance of concordances.

\begin{proposition}[\bf Gauge equivalences from null-concordance of concordances]:
1)In equivariant twistorial cohomotopy, the null concordances of concordances surject onto usual gauge transformations between the usual gauge potentials.
\begin{equation}
\label{Equitwistorial gauge equivalences from null concordances of concordances}
\renewcommand{\arraystretch}{1.3}
\begin{array}{l}
s_{2}|_{U_{i}^{\mathbb{Z}_{2}}}=\int_{s \in {[0,1]}} \int_{t \in {[0,1]}} \hat{\hat{p}}_{1}\\
	 B_{1}|_{U_{i}^{\mathbb{Z}_{2}}}=\int_{s \in{[0,1]}} \int_{t \in{[0,1]}} {\hat{\hat{H}}}_{3}\\A_{0}|_{U_{i}^{\mathbb{Z}_{2}}}= \int_{s \in{[0,1]}} \int_{t \in[0,1)} \hat{\hat{F}}_{2}\\

    \end{array}
\end{equation}
These are the results for the fixed locus, the bulk results are shown in \eqref{Twistorial gauge equivalnes from null concordances of concordances}.

2)The surjections from the null  concordance of concordances to the usual gauge transformations for the bulk $\mathrm{X}^{11}$ take the same form as in \eqref{Twistorial higher surjections}, and they take the following forms on the fixed locus $\mathrm{X}^{\mathbb{Z}_{2}}$:
\begin{equation}
    \label{Equitwistorial higher surjections}
    \renewcommand{\arraystretch}{1.6}
    \begin{array}{l}
    \hat{\hat{p}}_{1}=t p_{1}+\dd t s_{3}+s\dd t\left(s_{3}^{\prime}-s_{3}\right)-\dd s\, \dd t\, s_{2}\\
        
    \hat{\hat{H}}_{3}= t H_{3}+ \dd t\, B_{2} + t(1-t) A_{1} F_{2}
    + s\dd t\, (B_{2}^{\prime}-B_{2})\\
    + s t (1-t) (A_{1}^{\prime}- A_{1}) F_{2} + \dd s\, (t-t^2) A_{0} F_{2}- \dd s\, \dd t\, B_{1}\\
   \hat{\hat{F}}_{2}= tF_{2}+ \dd t A_{1} + s\dd t (A_{1}^{\prime}-A_{1})- \dd s\, \dd t\, A_{0}\,.
   \end{array}
    \end{equation}
    \end{proposition}
    
    These can be easily seen to satisfy the conditions in \eqref{Equitwistorial higher boundary} and also consistent with \eqref{Equitwistorial gauge tranfomrations}.

\section{Conclusion}
In this article, we have emphasized that flux quantization laws require the extra fields to make a theory globally defined. Motivated by the fact that the conjectured cohomology for eleven-dimensional supergravity is (twisted equivariant differential) cohomotopy, we have shown the modified Bianchi identities for M5 brane probes in 11D supergravity for tangentially twisted cohomotopy, twistorial cohomotopy and for equivariant twistorial cohomotopy. We have proved in these cases that the null concordances surject onto the usual differential forms of gauge potentials. Also, the concordances between the null concordances serve to provide the gauge transformations between the gauge potentials via surjections, thus providing the extra field content required to completely understand the {\it global} picture of M5 branes in 11D supergravity for tangentially twisted cohomotopy, twistorial cohomotopy and equivariant twistorial cohomotopy.


\begin{thebibliography}{10}
 
\bibitem{W-95}
E. Witten,
{\it \color{darkblue} String theory dynamics in various dimensions},
 Nucl. Phys. B {\bf 443} (1995), 85-126,
 [\href{https://doi.org/10.1201/9781482268737-32}{\tt doi:10.1201/9781482268737-32}],  
 [\href{https://arxiv.org/abs/hep-th/9503124}{\tt arXiv:9503124}].

\bibitem{U}
C. M. Hull, P. K. Townsend,
{\it \color{darkblue} Unity of Superstring Dualities},
Nucl.Phys.B438:109-137,1995,
 [\href{https://doi.org/10.1016/0550-3213%2894%2900559-W}{\tt doi:10.1016/0550-3213}],  
 [\href{https://arxiv.org/abs/hep-th/9410167v2}{\tt arXiv:9410167v2}].

\bibitem{4}
 P. K. Townsend,
{\it \color{darkblue} Four Lectures on M-theory},  
 [\href{https://arxiv.org/abs/hep-th/9612121v3}{\tt arXiv:9612121v3}].
\bibitem{Duff96}
M. Duff, 
{\it \color{darkblue} M-Theory (the Theory Formerly Known as Strings)}, Int. J. Mod. Phys. A {\bf 11} (1996), 5623-5642,  
[\href{https://doi.org/10.1142/S0217751X96002583}{\tt doi:10.1142/S0217751X96002583}],
[\href{https://arxiv.org/abs/hep-th/9608117}{\tt arXiv:hep-th/9608117}].


\bibitem{Duff99}
M. Duff, 
{\it \color{darkblue} The World in Eleven Dimensions: Supergravity, Supermembranes and M-theory}, IoP Publishing (1999), 
[\href{https://www.crcpress.com/The-World-in-Eleven-Dimensions-Supergravity-supermembranes-and-M-theory/Duff/9780750306720}{\tt ISBN:9780750306720}].

\bibitem{BST}
E. Bergshoeff(ICTP, Trieste), E. Sezgin(ICTP, Trieste), P.K. Townsend(Cambridge U.), 
{\it \color{darkblue} Supermembranes and Eleven-Dimensional Supergravity}, Phys.Lett.B 189 (1987) 75-78,  
[\href{https://doi.org/10.1201/9781482268737-10}{\tt doi: 10.1201/9781482268737-10}].
 
\bibitem{MS-05}
A. Miemiec, I. Schnakenburg,
{\it \color{darkblue} Basics of M-Theory},
 Fortsch.Phys. {\bf 54} (2006) 5-72, \newline 
 [\href{https://doi.org/10.1002/prop.200510256}{\tt doi:10.1002/prop.200510256}], 
 [\href{https://arxiv.org/abs/hep-th/0509137}{\tt arXiv:0509137}].

 \bibitem{GSS24-SuGra}
G. Giotopoulos, H. Sati, and U. Schreiber, {\it \color{darkblue} Super-Flux Quantization on 11d Superspace},
J. High Energy Physics {\bf 2024}  (2024) 82, 
[\href{https://doi.org/10.1007/JHEP07(2024)082}{\tt doi;10.1007/JHEP07(2024)082}],
[\href{https://arxiv.org/abs/2403.16456}{\tt arXiv:2403.16456}].

\bibitem{FSS23-Char}
D. Fiorenza, H. Sati, and U. Schreiber, 
{\it \color{darkblue} The Character map in Nonabelian Cohomology --- Twisted, Differential and Generalized},
World Scientific, Singapore (2023),
[\href{https://doi.org/10.1142/13422}{\tt doi:10.1142/13422}],
[\href{https://arxiv.org/abs/2009.11909}{\tt arXiv:2009.11909}].

\bibitem{SS25-Flux}
H. Sati, U. Schreiber,
{\it \color{darkblue} Flux Quantization},
Encyclopedia of Mathematical Physics 2nd ed. {\bf 4} Elsevier (2025), 281-324,
[\href{https://doi.org/10.1016/B978-0-323-95703-8.00078-1}{\tt doi:10.1016/B978-0-323-95703-8.00078-1}],
[\href{https://arxiv.org/abs/2406.11304}{\tt arXiv:2406.11304}]

\bibitem{Wi-Flux}
E. Witten,
{\it \color{darkblue}On Flux Quantization In M-Theory And The Effective Action},
J.Geom.Phys.22:1-13 (1997),
[\href{https://doi.org/10.1016/S0393-0440(96)00042-3}{\tt doi:10.1016/B978-S0-393-0440(96)00042-3}],
[\href{https://arxiv.org/abs/hep-th/9609122v2}{\tt arXiv:9609122v2}].

\bibitem{Sa-13}
H. Sati,
{\it \color{darkblue} Framed M-branes, corners, and topological invariants},
 J. Math. Phys. {\bf 59} (2018), 062304,
 [\href{https://doi.org/10.1063/1.5007185}{\tt doi:10.1063/1.5007185}], 
 [\href{https://arxiv.org/abs/1310.1060v2}{\tt arXiv:1310.1060v2}].

 \bibitem{2D}
Pavel Suprun,
{\it \color{darkblue} Towards 2-dimensional non-commutative integrals
},
 Phys. Rev. D 111, 025023, 2025,
 [\href{https://doi.org/10.1103/PhysRevD.111.025023}{\tt doi:10.1103}], 
 [\href{https://arxiv.org/abs/2406.19324v2}{\tt arXiv:2406.19324v2}].
 
 \bibitem{SS20-Equitwistorial}
H. Sati, U. Schreiber,
{\it \color{darkblue} The Character Map in Equivariant Twistorial Cohomotopy},
[\href{https://doi.org/10.48550/arXiv.2011.06533}{\tt doi:10.48550}],
[\href{https://arxiv.org/abs/2011.06533}{\tt arXiv:2011.06533}].


\bibitem{Guven}
 R. Gueven, {\it \color{darkblue} Black p-brane solutions of D = 11 supergravity theory}, Physics Letters B. 276 (49)  (1992),
 [\href{https://doi.org/10.1201/9781482268737-16}{\tt doi:10.1201/9781482268737-16}].

 \bibitem{Wi97}
E. Witten , {\it \color{darkblue} Five-brane effective action in M-theory}, Journal of Geometry and Physics (1997), 22 (2): 103–133, 
[\href{https://doi.org/10.1016/S0393-0440(97)80160-X}
 {\tt doi:10.1016/S0393-0440(97)80160-X}], [\href{https://arxiv.org/abs/hep-th/9610234v1}{\tt arXiv:9610234v1}] .

\bibitem{FKPZ}
 S. Ferrara, A. Kehagias, H. Partouche, A. Zaffaroni, {\it \color{darkblue} Membranes and Fivebranes with Lower Supersymmetry and their AdS Supergravity Duals}, Physics Letters B. 431 (1998): 42–48,
[\href{https://doi.org/10.1016/S0370-2693(98)00558-9}{\tt
 doi:10.1016/S0370-2693(98)00558-9}], [\href{https://arxiv.org/abs/hep-th/9803109v1}{\tt arXiv:9803109v1}] .

 
 \bibitem{PST}
P. Pasti, D. Sorokin, M. Tonin, {\it \color{darkblue}Covariant Action for a D = 11 Five-Brane with the Chiral Field}, Physics Letters B. 398 (1997): 41,
[\href{https://doi.org/10.1016/S0370-2693(97)00188-3}{\tt
doi:10.1016/S0370-2693(97)00188-3}],[\href{https://arxiv.org/abs/hep-th/9701037v3}{\tt arXiv:9701037v3}].

\bibitem{Berman}
David S. Berman,
{\it \color{darkblue}M-theory branes and their interactions},
Phys. Rept.456 (2008), 89-126,
[\href{https://doi.org/10.1016/j.physrep.2007.10.002}{\tt doi:2007.10.1016}],
[\href{https://arxiv.org/abs/0710.1707v2}{\tt arXiv:0710.1707v2}].

\bibitem{GSS24-FluxOnM5}
G. Giotopoulos, H. Sati, and U.  Schreiber,
{\it \color{darkblue} Flux Quantization on M5-Branes},
J. High Energy Physics {\bf 2024} (2024) 140,
[\href{https://doi.org/10.1007/JHEP10(2024)140}{\tt doi;10.1007/JHEP10(2024)140}],
[\href{https://arxiv.org/abs/2406.11304}{\tt arXiv:2406.11304}].

\bibitem{FSS20-H}
D. Fiorenza, H. Sati, and  U. Schreiber, 
{\it \color{darkblue} Twisted Cohomotopy implies M-theory anomaly cancellation on 8-manifolds}, 
Commun. Math. Phys. {\bf 377} (2020), 1961-2025, 
[\href{https://arxiv.org/abs/1904.10207}{\tt arXiv:1904.10207}],
[\href{https://doi.org/10.1007/s00220-020-03707-2}{\tt doi:10.1007/s00220-020-03707-2}].


\bibitem{FSS22-GS}
D. Fiorenza, H. Sati, and U. Schreiber,
{\it \color{darkblue} Twistorial Cohomotopy Implies Green-Schwarz anomaly cancellation},
Rev. Math. Phys.
{\bf 34} 05 (2022) 2250013,
[\href{https://doi.org/10.1142/S0129055X22500131}{\tt doi:10.1142/S0129055X22500131}],
[\href{https://arxiv.org/abs/2008.08544}{\tt arXiv:2008.08544}].


\bibitem{A-85}
O. Alvarez,
{\it \color{darkblue} Topological quantization and cohomology},
Commun. Math. Phys. {\bf 100} (1985), 279-309,
 [\href{https://doi.org/10.1007/BF01212452}{\tt doi:10.1007/BF01212452}].


\bibitem{FSS21-Hopf}
D. Fiorenza, H. Sati, and U. Schreiber,
{\it \color{darkblue} Twisted Cohomotopy implies M5 WZ term level quantization},
Commun. Math. Phys. {\bf 384} (2021), 403-432,
[\href{https://doi.org/10.1007/s00220-021-03951-0}{\tt doi:10.1007/s00220-021-03951-0}],
[\href{https://arxiv.org/abs/1906.07417}{\tt arXiv:1906.07417}].

\end{thebibliography}
\end{document}